\documentclass[12pt,a4paper]{article}
\usepackage[top=2.1cm,bottom=2.1cm,left=2.1cm,right=2.1cm]{geometry}
\usepackage{latexsym,amsmath,amsfonts,amssymb,amsthm,amscd,mathrsfs}

\usepackage{accents}

 \usepackage[unicode,colorlinks,plainpages=false,hyperindex=true,bookmarksnumbered=true,bookmarksopen=false,pdfpagelabels]{hyperref}
 \hypersetup{urlcolor=cyan,linkcolor=blue,citecolor=red,colorlinks=true}

\newcommand{\be}{\begin{equation}}
\newcommand{\ee}{\end{equation}}
\newcommand{\bea}{\begin{eqnarray}}
\newcommand{\eea}{\end{eqnarray}}

\numberwithin{equation}{section}
\newcounter{thmcounter}
\numberwithin{thmcounter}{section}

\theoremstyle{definition}
\newtheorem*{acknowledgements}{Acknowledgements}

\newtheorem{remark}[thmcounter]{Remark}
\theoremstyle{plain}
\newtheorem{corollary}[thmcounter]{Corollary}
\newtheorem{lemma}[thmcounter]{Lemma}
\newtheorem{proposition}[thmcounter]{Proposition}
\newtheorem{theorem}[thmcounter]{Theorem}

\def\C{\mathbb{C}}                          %
\def\N{\mathbb{N}}                          %
\def\T{\mathbb{T}}                          %
\def\reg{\mathrm{reg}}                      %
\def\dt {\left.\frac{d}{dt}\right|_{t=0}}   %
\def\fM{\mathfrak{M}}                       %
\def\GL{{\rm GL}(n,\C)}                     %
\def\gl{{\rm gl}(n,\C)}                     %
\def\cG{{\mathcal G}}                       %
\def\cR{{\mathcal R}}                       %
\def\half{\frac{1}{2}}                      %
\def\ad{{\mathrm{ad}}}                      %
\def\cO{{\cal O}}                           %
\def\cR{{\cal R}}                           %
\def\bR{{\mathbb R}}                        %
\def\1{{\mbox{\boldmath $1$}}}              %
\def\tr{\mathrm{tr\,}}                      %
\def\diag{\mathrm{diag}}                    %
\def\red{\mathrm{red}}                      %
\def\cL{{\cal L}}                           %
\def\cD{\mathcal{D}}                        %
\def\cG{\mathcal{G}}                        %
\def\cT{\mathcal{T}}                        %
\def\U{{\mathrm U}(n)}                      %
\def\ri{{\mathrm i}}                        %
\def\r{{\mathrm r}}                         %
\def\fA{\mathfrak{A}}                       %
\def\fH{\mathfrak{H}}                       %
\def\u{{\mathfrak u}}                       %
\def\cE{{\mathcal E}}                       %
\def\cZ{{\mathcal Z}}                       %
\def\cV{{\mathcal V}}                       %
\def\cF{{\mathcal F}}                       %
\def\cH{{\mathcal H}}                       %
\def\cU{{\mathcal U}}                       %
\def\cJ{{\mathcal J}}                       %
\def\cG{{\mathcal G}}                       %
\def\R{{\mathbb R}}                         %
\def\Q{\hat q}                              %

\begin{document}

\begin{center}
{\Large\bf
Bi-Hamiltonian structure of  Sutherland  models coupled to two  $\u(n)^*$-valued spins
from Poisson reduction}
\end{center}

\medskip
\begin{center}
L.~Feh\'er${}^{a,b}$
\\

\bigskip
${}^a$Department of Theoretical Physics, University of Szeged\\
Tisza Lajos krt 84-86, H-6720 Szeged, Hungary\\
e-mail: lfeher@physx.u-szeged.hu

\medskip
${}^b$Department of Theoretical Physics, WIGNER RCP, RMKI\\
H-1525 Budapest, P.O.B.~49, Hungary\\
\end{center}

\medskip
\begin{abstract}
We introduce a bi-Hamiltonian hierarchy on the cotangent bundle of the \emph{real} Lie group $\GL$,
 and study its Poisson reduction with respect
to the action of the product group $\U \times \U$ arising from left- and
right-multiplications.
One of the pertinent Poisson structures is the canonical one, while the other
is suitably transferred from the real Heisenberg double of $\GL$.
When taking the quotient of $T^* \GL$
we focus on the dense open subset of $\GL$ whose elements have  pairwise distinct singular values.
We develop a convenient description of the Poisson algebras of the $\U \times \U$ invariant
functions, and show that one of the Hamiltonians of the reduced bi-Hamiltonian hierarchy yields
a hyperbolic Sutherland
model coupled to two $\u(n)^*$-valued spins.
Thus we obtain a new bi-Hamiltonian interpretation  of this  model, which
represents a special case of Sutherland models coupled to two spins
obtained earlier
from reductions of cotangent bundles of reductive Lie groups equipped
with their canonical Poisson structure.
Upon setting one of the spins to zero, we
recover the bi-Hamiltonian structure of the standard hyperbolic spin Sutherland model
that was derived recently by a different method.
\end{abstract}

{\linespread{0.8}\tableofcontents}

\newpage
\section{Introduction and background}
\setcounter{equation}{0}

The investigation of integrable many-body models of Calogero--Moser--Sutherland type has been
initiated half a century ago \cite{Cal,Mo,S}, and  by now an extensive body of knowledge
has been amassed about such systems and their manifold applications in  mathematics and physics
 \cite{A,vDV,E,RBanff}.
Remarkably, we are still witnessing intense  research in this area, regarding especially
the `relativistic' deformations \cite{RS} of the
originally studied particle systems and their spin extensions \cite{FeNP,GH,KZ,LX}.
Recent enquiries on these models are concerned, for example, with their derivation by Hamiltonian
reduction \cite{AO,CF1,FFM},  relations
to moduli spaces of flat
connections \cite{AR}, harmonic analysis and special functions
 \cite{vDG,LNS,ReSt}, pole dynamics in  the matrix KP hierarchy \cite{PZ}
 and quiver varieties  \cite{CF1,Fair}.

Interestingly, several classical integrable systems admit
a bi-Hamiltonian structure, which means that the evolution equation can be encoded
by two different Hamiltonians and corresponding Poisson brackets that are
compatible in the sense that their linear combinations also satisfy the Jacobi identity.
This can often be used to generate a family of commuting Hamiltonians through recursion.
The first example of a bi-Hamiltonian structure was found by Magri \cite{M}  for the KdV equation.
Since then the bi-Hamiltonian approach has become one of the central tools in the study of soliton equations
 \cite{DeKV, MCFP}.
The bi-Hamiltonian structures of Toda lattices also received considerable
attention \cite{SurB},
but (except for the simplest rational model \cite{AAJ,Bar,FaMe,MCFP}) the construction of bi-Hamiltonian
structures for Calogero--Moser--Sutherland type models has been left largely unexplored
in the literature.  The aim of our current research is to take some steps towards filling this gap.

The present report is a continuation of our series of papers \cite{FeNon, FeLMP, FeAHP} devoted to the study of bi-Hamiltonian
aspects of spin Sutherland models. The papers \cite{FeNon,FeLMP}  dealt with bi-Hamiltonian structures of the
hyperbolic and trigonometric real spin Sutherland models, and then in \cite{FeAHP} we investigated the corresponding
complex holomorphic system.  Here, we construct
a bi-Hamiltonian structure for a more general real Sutherland type model that involves two spin variables
belonging to the dual $\u(n)^*$ of the Lie algebra of the unitary group $\U$.
The bi-Hamiltonian structure of the hyperbolic spin Sutherland model
will be recovered by setting one of the spins to zero.
The model of our interest is a special case of models investigated previously in
\cite{Feproc,FP2,KLOZ,Res3}, but in those papers bi-Hamiltonian structures were not addressed,
while this is the  issue that concerns us here.

Our work fits into the Hamiltonian reduction approach to integrable many-body models.
We now overview the chain of developments that motivated us and sketch the content
of the present report.
The first step was taken by Olshanetsky and Perelomov \cite{OP} who derived the hyperbolic Sutherland model
with Hamiltonian,
\be
\cH_{\mathrm{Suth}} =
\frac{1}{2}\sum_{i=1}^n  p_i^2  + \sum_{1\leq  i<j\leq n}\frac{\kappa^2 }{\sinh^2(q_i- q_j)},   \qquad \kappa\in \R^*,
\label{Suth}\ee
by reduction of geodesic motion on the Riemannian symmetric space of positive matrices.
One may view the symmetric space in question as
the coset space $G/\U$ for $G:= \GL$
regarded as a real Lie group.
The pioneering work  \cite{OP} focused on solving the equations of motion defined by the Hamiltonian \eqref{Suth}
via projection of suitably constrained geodesics on $G/\U$.
It is fruitful to interpret this as Hamiltonian reduction
of the phase space $T^*G$ as follows \cite{FK1}.
There is a natural action of the group $\U \times \U$
on $T^* G$ (see equations  \eqref{cR} and \eqref{cL} below), which is generated by a moment map
\be
(\Phi_L, \Phi_R): T^* G \to \u(n)^* \oplus \u(n)^*.
\label{Phi1}\ee
The Sutherland model lives on the reduced phase space
\be
\Phi_L^{-1}(0) \cap \Phi_R^{-1}(\cO_{\mathrm {KKS}})/ (\U \times \U),
\label{KKS}\ee
where $\cO_{\mathrm{KKS}}$ (named after \cite{KKS}) is a coadjoint orbit of $\U$ of dimension $2(n-1)$ (that is,
the smallest non-trivial coadjoint orbit).
One can also perform the reduction in two consecutive steps, descending first to $\Phi_L^{-1}(0)/\U$,
which is nothing but the cotangent bundle of the symmetric space.
In this construction it is useful to trivialize $T^*G$, say by right translations, and also identify $\cG:=\gl$ with its own
dual space by means of the pairing given by the real part of the trace form, see \eqref{E1}.
Thus $T^*G$ gets identified with the manifold
\be
\fM: = \{ (g,J) \mid g\in G, \, J\in \cG\},
\ee
and $\cH_{\mathrm{Suth}}$ \eqref{Suth} arises from the unreduced Hamiltonian
\be
H(g,J) = \frac{1}{2} \langle J, J \rangle.
\label{freeHam}\ee
Of course, the role of the two $\U$ factors can be exchanged in \eqref{KKS}.

A generalization \cite{FP1} of the construction
just outlined is to consider the reduction to
\be
\Phi_L^{-1}(0)/ (\U \times \U).
\label{NPBcase}\ee
After restriction to a dense open subset, this gives rise to the spin Sutherland model with
Hamiltonian
\be
\cH_{\mathrm{spin-1}} =
\frac{1}{2}\sum_{i=1}^n p_i^2  + \sum_{1\leq  i<j\leq n}\frac{|\xi_{ij}|^2 }{\sinh^2(q_i- q_j)},
\label{spinSuth1}\ee
where the `spin variable' $\xi \in \u(n)^*\simeq \u(n)$ has zero diagonal part.
The spin $\xi$ can be restricted to a coadjoint orbit by imposing a corresponding constraint on $\Phi_R$, as was done in \cite{FP1},
but this will not be advantageous for our present purpose.

The most general reduction of $\fM$ based on the group $\U \times \U$
consists in descending to the quotient
\be
\fM/(\U \times \U),
\label{quotient}\ee
which is naturally a Poisson space.
Taking the quotient of a Poisson manifold by a group action, having the key property that
the Poisson bracket closes on the invariant functions, is usually called Poisson reduction.
The smooth functions on the quotient stem from the smooth invariant functions upstairs.
In the paper \cite{FP2} the closely related method of symplectic reduction was applied, which means that we
considered
\be
\Phi_L^{-1}(\cO_L) \times \Phi_R^{-1}(\cO_R)/ (\U \times \U)
\ee
with two arbitrary coadjoint orbits of $\U$.
By general principles \cite{OR}, these quotients represent
Poisson subspaces of the (singular) Poisson space \eqref{quotient}.
They are stable under the projections of the Hamiltonian
flows of all the $\U \times \U$ invariant functions.

The fact that the quotient space \eqref{quotient}
is not a smooth manifold does not cause any serious difficulty.
Still, it is convenient to focus on the dense open subset $\fM^\reg \subset \fM$, where the singular value decomposition of $g\in G$
has the form
\be
g = A e^q B^{-1}
\quad\hbox{with}\quad
q = \diag(q_1, q_2,\dots, q_n),\quad
q_i\in \R,\, q_1 > q_2 > \dots > q_n,
\label{sing}\ee
and $A,B\in \U$.  Note that
$q$ is uniquely determined by $g$, but $A$ and $B$ are unique only up to the transformations
\be
(A,B) \mapsto (A\tau, B\tau),
\qquad
\forall
\tau \in \T^n,
\label{ABtau}\ee
with $\T^n < \U$ denoting the diagonal subgroup.
Introduce the space
\be
\fM_0^\reg := \{ (e^q, \cJ)\},
\ee
where $q$ is restricted as in \eqref{sing} and $\cJ\in \cG$.  Then consider the mapping
\be
(g,J) \mapsto (e^q,\cJ)  \quad\hbox{with}\quad \cJ:= A^{-1} J A
\ee
from $\fM^\reg$ onto $\fM_0^\reg$.
This induces a bijective mapping of the dense subset $\fM^\reg/ (\U\times \U)$
of the quotient \eqref{quotient}  onto $\fM_0^\reg/\T^n$.
The quotient by $\T^n$  expresses that
$\cJ$ matters up to the gauge transformations $\cJ \mapsto \tau \cJ \tau^{-1}$, as inherited from \eqref{ABtau}.
According to earlier results (see equation (2.25) in \cite{FP2}, and see also \cite{KLOZ,Res3}),
one can parametrize the matrix  $\cJ$ in terms of two $\u(n)^*$-valued spin variables
$\xi^l$ and $\xi^r$ as follows:
\be
\cJ_{ij}= p_i \delta_{ij}   - (1-\delta_{ij})\left( \coth(q_i -q_j) \xi^l_{ij} + \xi^r_{ij}/\sinh(q_i-q_j)\right) - \xi^l_{ij}.
\label{cJpar}\ee
The spin variables Poisson commute with the canonical pairs $q_i, p_i$.
We may decompose any $n\times n$
 complex matrix $X$ into anti-Hermitian part $X^+$ and Hermitian part $X^-$, and then
we have $\xi^l = - \cJ^+$ and $\xi^r = (e^{-q} \cJ e^q)^+$.  The two spins are coupled by
the constraint that the diagonal part of $(\xi^l + \xi^r)$ vanishes, and the $\T^n$ gauge transformations
act on them by simultaneous conjugations.
The reduction of the Hamiltonian \eqref{freeHam} takes the form
\be
\cH_{\mathrm{spin-2}}
=\frac{1}{2}\sum_{i=1}^n (p_i^2  - \vert \xi^l_{ii}\vert^2) + \sum_{ 1\leq i<j\leq n}\Biggl(\frac{|\xi_{ij}^l|^2 + |\xi^r_{ij}|^2
- 2 \Re(\xi^r_{ij} \xi^l_{ji})}{\sinh^2(q_i- q_j)}
+\frac{\Re(\xi^r_{ij} \xi^l_{ji})}{\sinh^2((q_i- q_j)/2)}\Biggr).
\label{spinSuth2}\ee
Upon setting  $\xi^l=0$, $\cJ$ simplifies to the Hermitian Lax matrix of the spin Sutherland model
and
 $\cH_{\mathrm{spin-2}}$ turns into the spin Sutherland Hamiltonian \eqref{spinSuth1}.

The essential new contribution of the present paper is this. We first exhibit an interesting
bi-Hamiltonian structure on the cotangent bundle $T^*G$. This consists of a quadratic Poisson
bracket (in terms of natural coordinate functions) and the canonical one.
The quadratic Poisson bracket comes from the theory of Poisson--Lie groups \cite{STS}, and its
 suitable Lie derivative
is the canonical Poisson bracket. The Hamiltonians given by the real and imaginary
parts of $\tr(J^k)$, for $k\in \N$,  generate a hierarchy of bi-Hamiltonian evolution equations on $\fM$.
The Poisson reduction with respect to $\U \times \U$ gives rise to a bi-Hamiltonian structure on the quotient space \eqref{quotient},
 which
we will realize in the form of compatible Poisson brackets on the space of $\T^n$ invariant functions
$C^\infty(\fM_0^\reg)^{\T^n}$.
These functions represent the ring of smooth functions on the regular part of the quotient \eqref{quotient}.
In practice, this means that we present the reduced Poisson brackets in terms of $\T^n$ invariant functions
of the pair $(q,\cJ)$.
By constraining the anti-Hermitian part of $\cJ$ to zero, we reach a joint Poisson subspace
of the compatible Poisson brackets, and on this subspace
we recover the bi-Hamiltonian structure of the spin Sutherland model \eqref{spinSuth1} that was
derived originally by a rather ad hoc method \cite{FeNon}.
In the general case,  we shall explain
that in terms of the variables $(q,p,\xi^l, \xi^r)$ the first reduced Poisson structure
takes the form that underlies the interpretation
of the Hamiltonian \eqref{spinSuth2} as a Sutherland model coupled to two spins.

The article is organized as follows. In Section 2, we describe the compatible Poisson brackets
and a family of bi-Hamiltonian evolution equations on $\fM$.
In Section 3, we develop the reduction of these structures.
Section 4 is devoted to the spin Sutherland interpretation of the reduced system.
Our conclusions are given in Section 5. There are also three appendices.
Appendix A provides an explanation of the origin of
the second Poisson bracket on $\fM$, Appendix B expounds some technical details,
and  Appendix C contains the proof of the key property of the
change of variables behind the spin Sutherland interpretation.

Finally, let us state explicitly what we consider as our new results and
specify their location in the text.
First, we think that the unreduced bi-Hamiltonian hierarchy that we present in Section 2
has not been described before. We studied the analogous holomorphic system in \cite{FeAHP}.
Our principal results are contained in Section 3: Theorems \ref{Prop:3.4} and  \ref{Prop:3.5} give the reduced Poisson brackets
in terms of $\T^n$ invariant functions on $\fM_0^\reg$, Propositions \ref{Th:3.6} and \ref{Prop:3.8} characterize
the reduced bi-Hamiltonian hierarchy of evolution equations.
In the general case,  the interpretation as a spin Sutherland model with two spins is explained
by Lemma \ref{lem:4.6} and Proposition \ref{Prop:4.8}, which were motivated by previous results \cite{FP2}.
Theorem \ref{Th:4.3}  represents an important new result, since it clarifies the origin
of the bi-Hamiltonian structure of the hyperbolic spin Sutherland model having
the `main Hamiltonian' \eqref{spinSuth1}.

\begin{remark} \label{Rem1}
To be more precise about the related literature \cite{Feproc,FP2,KLOZ,Res3}, we note that
\cite{Feproc} and \cite{Res3} considered generalized spin Sutherland models
on symplectic leaves of quotient spaces of the form
 $K_1\backslash T^*G/K_2$, where
$K_1$ and $K_2$ are the fixed point subgroups  of two involutions of a semisimple or reductive Lie group $G$,
while in \cite{FP2,KLOZ} $K_1=K_2$ was taken to be the maximal compact subgroup.
Our construction of the bi-Hamiltonian master system on $T^*G$ is specific to $G=\GL$, since it uses
that $\GL$ is an open subset of its Lie algebra, which contains the unit matrix.
One may also apply a similar construction to $\mathrm{GL}(n,\R)$, but no other generalizations
are apparent to us.
\end{remark}

\section{Master system on the
cotangent bundle of $\GL$}
\setcounter{equation}{0}

We first fix some notations.
Let us consider  $G:= \GL$, regarded as a \emph{real} Lie group, and
endow its Lie algebra $\cG:= \gl$ with the trace form
\be
\langle X, Y \rangle := \Re\tr(XY),
\quad \forall X, Y \in \cG.
\label{E1}\ee
Any $X\in \cG$ admits the unique decomposition
\be
X = X_> + X_0 + X_<
\label{E2}\ee
into strictly upper triangular part $X_>$, diagonal part $X_0$, and strictly lower triangular part $X_<$.
Correspondingly, one obtains the vector space direct sum
\be
\cG = \cG_{>} + \cG_0 + \cG_<.
\label{cGdec1}\ee
This decomposition gives rise to the standard solution of the modified classical Yang--Baxter equation
on $\cG$,  $r\in \mathrm{End}(\cG)$. Specifically, we define
\be
r(X):=   \frac{1}{2} (X_> - X_<), \qquad \forall X\in \cG,
\label{E3}\ee
and also introduce
\be
r_\pm := r \pm \frac{1}{2} \mathrm{id},
\label{rpm}\ee
where $\mathrm{id}(X):= X$.  Note that
\be
\langle r(X), Y \rangle = - \langle X, r(Y) \rangle, \qquad \forall X,Y\in \cG.
\label{asym}\ee

Let  $\cG^+$ denote the subalgebra of anti-Hermitian matrices and $\cG^-$ the subspace of Hermitian matrices,
\be
\cG^+ :=\u(n) = \{ X \in \cG \mid X^* = - X\},
\qquad
\cG^-:= \ri \u(n).
\label{Gpm}\ee
They yield another direct sum decomposition
\be
\cG = \cG^+ + \cG^-,
\label{ortog}\ee
where the two subspaces are orthogonal to each other with respect to the bilinear form \eqref{E1}.
Thus we may decompose any $X\in \cG$ as $X= X^+ + X^-$ with $X^\pm \in \cG^\pm$.
We shall use that
\be
[\cG^+, \cG^+] \subset \cG^+,\quad
[\cG^-, \cG^-] \subset \cG^+,
\quad
[\cG^+, \cG^-] \subset \cG^-,
\label{comrel}\ee
and, in consequence of \eqref{E3},
\be
r(\cG^+) \subset \cG^-
\quad\hbox{and}\quad r(\cG^-)\subset \cG^+.
\label{rmap}\ee

Our immediate aim is to present two Poisson structures on  the smooth, real manifold
\be
\fM:= G \times \cG = \{ (g,J)\mid g\in G, \, J\in \cG\}.
\label{E4}\ee
We view $\fM$  as a model of the cotangent bundle of $T^*G$.
Here, the cotangent bundle is trivialized by means of right translations,
and $\cG^*$ is identified with $\cG$ using the trace form.

For any $F \in C^\infty(\fM,\bR)$, introduce the $\cG$-valued derivatives $\nabla_1 F$, $\nabla_1' F$
 and $d_2 F$ by the defining relations
\be
\langle \nabla_1 F(g,J), X\rangle = \dt F(e^{tX} g, J),
\quad
\langle \nabla_1' F(g,J), X\rangle = \dt F(ge^{tX}, J),
\ee
and
\be
\langle d_2 F(g,J), X \rangle = \dt F(g, J + t X),
\label{E5}\ee
where $t$ is a real parameter and $X\in \cG$ is arbitrary.
In addition, it will be convenient to define the $\cG$-valued
functions $\nabla_2 F$ and $\nabla_2' F$ by
\be
\nabla_2 F(g,J) := J (d_2 F(g,J)),
\qquad
\nabla_2' F(g,J) := (d_2F(g,J)) J.
\label{E7}\ee
It is worth remarking that
$\nabla_1 F(g,J) = g \left(\nabla'_1 F(g,J) \right)g^{-1}$.

\medskip\noindent
\begin{theorem} \label{Th:2.1}
For $F, H \in C^\infty(\fM,\R)$, the following formulae define two
Poisson brackets:
\be
\{ F,H\}_1(g,J) =   \langle \nabla_1 F, d_2 H\rangle - \langle \nabla_1 H, d_2 F \rangle +
\langle J, [d_2 F, d_2 H]\rangle,
\label{PB1}\ee
and
\bea
\{ F, H\}_2(g,J)  &=&
 \langle r \nabla_1 F, \nabla_1 H \rangle - \langle r \nabla'_1 F, \nabla'_1 H \rangle
\label{PB2} \\
&&
 +\langle \nabla_2 F - \nabla_2' F, r_+\nabla_2' H  - r_- \nabla_2 H
   \rangle \nonumber \\
 &&
 +\langle \nabla_1 F,  r_+ \nabla_2' H  - r_-\nabla_2 H  \rangle
- \langle \nabla_1 H,  r_+ \nabla_2' F  - r_- \nabla_2 F \rangle,
\nonumber\eea
where the derivatives are evaluated at $(g,J)$, and we put $rX$ for $r(X)$.
\end{theorem}

\begin{proof}
The first Poisson bracket is easily recognized to be the canonical one carried by
the cotangent bundle $T^*G$, presented by means of right trivialization.
The second one is the real analogue of the holomorphic Poisson bracket of $\fM$, then
regarded as a complex manifold, which has been constructed in \cite{FeAHP}.
 The formula \eqref{PB2}  has the same form as the second Poisson bracket in \cite{FeAHP},
 except that now we use smooth real functions and the real part of the trace for the pairing
$\langle\ ,\ \rangle$.  We give a terse outline of the origin of this Poisson bracket
in Appendix A. Of course, its Jacobi identity can be verified directly as well.
\end{proof}

Let us recall \cite{M,MCFP} that two Poisson brackets are called \emph{compatible} if their arbitrary linear
combination is also a Poisson bracket. A \emph{bi-Hamiltonian manifold} is a manifold
equipped with a pair of compatible Poisson brackets.
The non-trivial condition is posed by the Jacobi identity for the linear combination
of the Poisson brackets. In several examples it is guaranteed by the following mechanism.
Suppose that $\cD$ is a derivation of the functions, and $\{\ ,\ \}$ is a Poisson bracket.
Then the formula
\be
\{F,H\}^\cD := \cD[\{F,H\}] -  \{ \cD[F], H\} - \{ F, \cD[H]\},
\ee
where $\cD[F]$ is the derivative of the function $F$, provides an anti-symmetric bi-derivation.
This is called the Lie derivative bracket (since its underlying bi-vector is
the Lie derivative of the Poisson tensor of $\{\ ,\ \}$ along the corresponding vector field).
It is well-known \cite{Sm} that if $\{\ ,\ \}^\cD$ satisfies the Jacobi identity, then
this holds also for arbitrary linear combinations of $\{\ ,\ \}$ and $\{\ ,\ \}^\cD$.

Now consider the vector field on $\fM$ whose flow through $(g,J)$ is
\be
(g, J + t \1_n),
\qquad
t \in \bR,
\ee
where $\1_n$ denotes the $n\times n$ unit matrix,
and let $\cD$ be the corresponding derivation of $C^\infty(\fM,\R)$.
As coordinate functions on $\fM$, we may employ  `evaluation functions'
that on $(g,J)$ return the real and imaginary parts of the matrix elements of $g$,
\be
g_{\alpha \beta}^\r := \Re g_{\alpha \beta},
\quad
 g_{\alpha \beta}^\ri := \Im g_{\alpha \beta},
\label{gcomp}\ee
and the components $J_k:= \langle T_k, J\rangle$ with respect to an arbitrary
 basis $T_k$ ($k=1,\dots, 2 n^2$) of the real vector space $\cG$.
Their derivatives are
\be
\cD[g_{\alpha \beta}] =0,
\qquad
\cD[J_k] = \langle T_k, \1_n\rangle,
\ee
where we extended the derivation to complex functions.

\medskip\noindent
\begin{proposition} \label{Prop:2.2}
The first Poisson bracket \eqref{PB1} on $\fM$ is the Lie derivative
of the second Poisson bracket \eqref{PB2}  with
respect to the derivation $\cD$, i.e., we have
\be
\{F,H\}_1 = \{F,H\}_2^\cD,
\qquad
\forall F,H \in C^\infty(\fM,\R).
\label{claim1}\ee
Consequently, the two Poisson brackets of Theorem 2.1 are compatible.
\end{proposition}

\begin{proof}
Since both $\{\ ,\ \}_1$ and $\{\ ,\ \}_2^\cD$ are bi-derivations, it is enough to verify the equality
\eqref{claim1} for coordinate functions. Hence, we inspect the cases for which $F$ and $H$ are taken
from the functions $J_k$ and $g_{\alpha\beta}^\ri$, $g_{\alpha \beta}^\r$.
Since $\cD[g_{\alpha \beta}]=0$ and the second Poisson bracket closes on the functions that depend only on $g$,
we see that the equality \eqref{claim1} holds if both $F$ and $H$ depend only on $g$, for in these cases
both sides of \eqref{claim1} vanish.
Next, we note that
\be
\{J_k, J_l\}_2 = \langle [J,T_k], r[T_l,J] + \frac{1}{2} (T_l J + J T_l) \rangle,
\label{Jkl}\ee
and this leads to
\be
\{J_k, J_l\}_2^\cD = \langle J, [T_k, T_l]\rangle,
\ee
which equals $\{J_k, J_l\}_1$. To continue, we introduce the $n\times n$ matrix
$\{ g, J_k\}_a$, for  $a=1,2$, by defining  its $\alpha \beta$ entry to be
\be
\{ g_{\alpha \beta}^r, J_k\}_a + \ri \{ g_{\alpha \beta}^\ri, J_k\}_a,
\qquad 1\leq \alpha, \beta\leq n.
\ee
The formulae of Theorem \ref{Th:2.1} give
\be
\{ g, J_k\}_1 = T_k g
\quad\hbox{and}\quad
\{ g, J_k\}_2 =\left( r [T_k, J] + \frac{1}{2} (J T_k + T_k J) \right) g
\ee
From here it is easily checked that
\be
\{ g_{\alpha \beta}^r, J_k\}_1 + \ri \{ g_{\alpha \beta}^\ri, J_k\}_1 =
\{ g_{\alpha \beta}^r, J_k\}_2^\cD + \ri \{ g_{\alpha \beta}^\ri, J_k\}_2^\cD,
\ee
whereby the proof is complete.
\end{proof}

The evolution equation of a \emph{bi-Hamiltonian system} can be written in Hamiltonian form with
respect to two compatible, linearly independent, Poisson brackets, and respective Hamiltonians.
We now present interesting bi-Hamiltonian systems on $\fM$.

\medskip\noindent
\begin{proposition} \label{Prop:2.3}
For any $k\in \N$, define the  Hamiltonians $H_k$ and $\tilde H_k$ on $\fM$ by
\be
H_k(g,J) := \frac{1}{k} \Re \tr(J^k),
\qquad
\tilde H_k(g,J) := \frac{1}{k} \Im \tr( J^k).
\label{HamsonfM}\ee
All these Hamiltonians are in involution, and they define bi-Hamiltonian
systems according to the relations, and corresponding flows, listed as follows:
\be
\{\,\cdot\,, H_k\}_2 = \{\,\cdot\,, H_{k+1}\}_1,
\qquad
(g(t), J(t)) = \left(\exp\left(J(0)^k t\right) g(0), J(0)\right)
\label{bihrel1}\ee
and
\be
\{\,\cdot\,, \tilde H_k\}_2 = \{\,\cdot\,, \tilde H_{k+1}\}_1,
\qquad
(g(t), J(t)) = \left(\exp\left(-\ri J(0)^k t\right) g(0), J(0)\right).
\label{bihrel2}\ee
  \end{proposition}
\begin{proof}
We can calculate that
\be
d_2 H_k = J^{k-1},\qquad  \quad d_2 \tilde H_k = - \ri J^{k-1}.
\label{d2Hams}\ee
Therefore $(\nabla'_2 H_k - \nabla_2 H_k)$ and  $(\nabla'_2 \tilde H_k - \nabla_2 \tilde H_k)$ both vanish, and we obtain
\be
 r_+ \nabla'_2 H_k - r_- \nabla_2 H_k = d_2 H_{k+1} = J^k,
\qquad
 r_+ \nabla'_2 \tilde H_k - r_- \nabla_2 \tilde H_k = d_2 \tilde H_{k+1} = -\ri J^k.
\ee
Taking an arbitrary $F\in C^\infty(\fM, \R)$, the formulae of Theorem \ref{Th:2.1} then yield
\be
\{ F, H_k\}_2(g,J) = \{ F, H_{k+1}\}_1(g,J) = \langle \nabla_1 F(g,J), J^k \rangle,
\ee
and
\be
\{ F, \tilde H_k\}_2(g,J) = \{ F, \tilde H_{k+1}\}_1(g,J) = \langle \nabla_1 F(g,J), -\ri J^k \rangle.
\ee
This means that the Poisson bracket relations \eqref{bihrel1} and \eqref{bihrel2} hold.
It follows also immediately that the corresponding Hamiltonian flows through the initial value $(g(0),J(0))$
have the simple form as claimed. In particular, $J$ and all its functions are constant along these flows,
which implies the involutivity of the underlying Hamiltonians.
\end{proof}

\medskip\noindent
\begin{remark} \label{Rem:2.4}
Let us define the $\cG$-valued function $\tilde J$ on $\fM$ by
\be
\tilde J(g,J):=- g^{-1} J g,
\label{tJdef}\ee
and notice that $\tilde J$ is constant along all the bi-Hamiltonian flows of Proposition \ref{Prop:2.3}.
The functions of $J$ and $\tilde J$ are related by
the $2n$ functional relations
\be
\Re \tr(J^k) = (-1)^k \Re\tr(\tilde J^k),
\quad
\Im \tr(J^k) = (-1)^k \Im\tr(\tilde J^k),
\qquad\forall k=1,\dots,n.
\label{funrels}\ee
Thus the functional dimension of the ring of functions, $\fA$,  depending on $J$ and $\tilde J$
equals $4 n^2 - 2n$.   This ring is closed under both Poisson brackets, since in addition to \eqref{Jkl} one
has $\{ J_k, \tilde J_l\}_i=0$ ($i=1,2$)  
and
\be
\{\tilde J_k, \tilde J_l\}_2 = -\langle [\tilde J,T_k], r[T_l,\tilde J] + \frac{1}{2} (T_l \tilde J + \tilde J T_l) \rangle.
\label{tJkl}\ee
The elements of $\fA$ Poisson commute with the elements of
the ring of functions, $\fH \subset \fA$, generated
 by the Hamiltonians \eqref{HamsonfM}.
  It is clear that the functional dimension of $\fH$ is $2n$.
Let us recall \cite{Nekh} that a \emph{ degenerate integrable system} on a symplectic manifold of dimension $N$
is given
by an Abelian Poisson algebra of functional dimension $N_1< N/2$, whose element admit $N-N_1$ functionally
independent joint
constants of motion\footnote{Degenerate integrable systems are also called
\emph{superintegrable}
or \emph{integrable in the non-commutative sense} \cite{LMV,MF,Res2}. Liouville integrability corresponds to the limiting case $N_1=N/2$. }.
We can directly apply this definition to the first Poisson structure, which is symplectic.
Therefore we see that the Hamiltonians \eqref{HamsonfM} define a degenerate integrable system,
 at least with respect the first Poisson bracket.
The origin of the second Poisson bracket from Poisson--Lie groups (see Appendix \ref{app:A}) shows
that it is also symplectic on a dense open subset of $\fM$, and therefore we may say that
the Hamiltonians of the `bi-Hamiltonian hierarchy' of Proposition \ref{Prop:2.3}
 form a degenerate integrable system
with respect to both Poisson brackets.
\end{remark}

 \section{Poisson reduction with respect to $\U\times \U$}
 \setcounter{equation}{0}

We are going to reduce with respect to an action of $\U\times \U$, which is engendered by two
commuting actions of $\U$, called right-handed and left-handed actions.
Regarding the first Poisson structure, it will be obvious that the Poisson bracket closes on the
invariant functions. For the second Poisson bracket, this needs a proof, and we start by
providing it.

We define the right-handed action of $\U$ on $\fM$ by associating with any $\eta \in \U$ the diffeomorphism
$\cR_\eta$ of $\fM$  that acts according to
\be
\cR_\eta: (g, J) \mapsto (g \eta^{-1}, J).
\label{cR}\ee
The left-handed action is given by $\cL_\eta$,
\be
\cL_\eta: (g,J) \mapsto (\eta g, \eta J \eta^{-1}).
\label{cL}\ee
These are commuting $\U$ actions on $\fM$, for we have
\be
\cR_{\eta_1} \circ \cR_{\eta_2}= \cR_{\eta_1 \eta_2},
\qquad
\cL_{\eta_1} \circ \cL_{\eta_2}= \cL_{\eta_1 \eta_2},
\qquad
\cL_{\eta_1} \circ \cR_{\eta_2}=\cR_{\eta_2} \circ \cL_{\eta_1},
\quad \forall \eta_1, \eta_2\in \U.
\ee
We shall consider the respective sets of invariant functions,
\be
C^\infty(\fM)^{\U}_\cR:= \{ F\in C^\infty(\fM,\R)\mid F \circ \cR_\eta =F, \quad \forall \eta\in \U\},
\label{inv1}\ee
\be
C^\infty(\fM)^{\U}_\cL:= \{ F\in C^\infty(\fM,\R)\mid F \circ \cL_\eta =F, \quad \forall \eta\in \U\},
\label{inv2}\ee
and
\be
C^\infty(\fM)^{\U \times \U} = C^\infty(\fM)^{\U}_\cR \cap C^\infty(\fM)^{\U}_\cL.
\label{inv3}\ee

\medskip\noindent
\begin{lemma} \label{Lem:3.1}
For $F,H\in C^\infty(\fM)^{\U}_\cR$,  $\{ F, H\}_2$  \eqref{PB2}  takes the form
\bea
\{ F, H\}_2  &=&
 \langle r \nabla_1 F, \nabla_1 H \rangle
 +\langle \nabla_2 F - \nabla_2' F, r_+\nabla_2' H  - r_- \nabla_2 H
   \rangle \nonumber \\
 &&
 +\langle \nabla_1 F,  r_+ \nabla_2' H  - r_-\nabla_2 H  \rangle
- \langle \nabla_1 H,  r_+ \nabla_2' F  - r_- \nabla_2 F \rangle\,,
\label{PB2R}\eea
and it belongs to $C^\infty(\fM)^{\U}_\cR$.
\end{lemma}

\begin{proof}
If $F\in C^\infty(\fM)^{\U}_\cR$, then
\be
F(g e^{tX}, J) = F(g,J),
\qquad
\forall X \in \cG^+,
\ee
and therefore
\be
\nabla'_1 F(g,J) \in \cG^-.
\label{rightinv}\ee
On account of  \eqref{rmap} and the orthogonality of the subspaces \eqref{ortog},
we see that for right-invariant functions $F,H$
\be
\langle r\nabla_1' F, \nabla_1' H \rangle = 0.
\label{310}\ee
Thus we obtain \eqref{PB2R} from \eqref{PB2}.
The closure of the Poisson bracket on $C^\infty(\fM)^{\U}_\cR$ is then a consequence of the relations
\be
(\nabla_1 F) \circ \cR_\eta = \nabla_1 F,
\quad
( \nabla_2 F) \circ \cR_\eta = \nabla_2 F,
\quad
( \nabla_2' F) \circ \cR_\eta = \nabla_2' F,
\ee
which follow directly from the definitions.
\end{proof}

\medskip\noindent
\begin{lemma} \label{Lem:3.2}
For $F,H\in C^\infty(\fM)^{\U}_\cL$,  $\{ F, H\}_2$  \eqref{PB2}
 takes the form
\bea
\{ F, H\}_2  &=&
\langle \nabla_1' F, r \nabla_1' H \rangle
 +\half \langle \nabla_2 F ,\nabla_2' H \rangle
 -\half \langle \nabla_2 H ,\nabla_2' F \rangle\nonumber \\
 &&
 +\half \langle \nabla_1 F,  \nabla_2' H  +\nabla_2 H  \rangle
- \half \langle \nabla_1 H,   \nabla_2' F + \nabla_2 F \rangle,
\label{PB2L}\eea
and it belongs to $C^\infty(\fM)^{\U}_\cL$.
\end{lemma}
\begin{proof}
The function $F$ (and similarly $H$) enjoys the invariance property
\be
F(e^{tX} g, e^{tX} J e^{-tX})  = F(g,J),
\qquad
\forall X\in \cG^+.
\ee
Taking the derivative at $t=0$ we obtain
\be
(\nabla_1 F)^+ = (\nabla_2'F - \nabla_2 F)^+.
\label{id1}\ee
Here and below, we apply the decomposition $Y= Y^+ + Y^-$ for any $Y\in \cG$, as defined
by the direct sum \eqref{ortog}.
We now look at the terms that appear in the formula \eqref{PB2}.
By using that $r$ maps $\cG^\pm$ into $\cG^\mp$ \eqref{rmap} and taking \eqref{id1} into account,  we derive the identity
\be
\langle r \nabla_1 F, \nabla_1 H\rangle =
\langle r (\nabla_1 H)^-,   (\nabla_2 F - \nabla_2' F)^+ \rangle
-\langle r (\nabla_1 F)^-,   (\nabla_2 H - \nabla_2' H)^+ \rangle.
\ee
In a similar manner, we obtain
\bea
&&\langle \nabla_1 F,  r \nabla_2' H  - r\nabla_2 H  \rangle
- \langle \nabla_1 H,  r \nabla_2' F  - r \nabla_2 F \rangle = \nonumber \\
&&\quad \langle r (\nabla_1 F)^-,   (\nabla_2 H - \nabla_2' H)^+ \rangle
-\langle r (\nabla_1 H)^-,   (\nabla_2 F - \nabla_2' F)^+ \rangle \\
&&\quad + \langle r (\nabla_2' F - \nabla_2 F)^+,   (\nabla_2 H - \nabla_2' H)^- \rangle
- \langle r (\nabla_2' H - \nabla_2 H)^+,   (\nabla_2 F - \nabla_2' F)^- \rangle.\nonumber
\eea
Next, we can write
\bea
&&\langle \nabla_2 F - \nabla_2' F, r \nabla_2' H  - r \nabla_2 H \rangle =
\\
&&\quad
\langle r (\nabla_2' F - \nabla_2 F)^+,   (\nabla_2' H - \nabla_2 H)^- \rangle
- \langle r (\nabla_2' H - \nabla_2 H)^+,   (\nabla_2' F - \nabla_2 F)^- \rangle. \nonumber
\eea
In these derivations we used the anti-symmetry of $r$ \eqref{asym}.
Observe that the terms given by the last three equations cancel altogether.
In other words, the terms that contain $r$ in $r_\pm$ \eqref{rpm} all cancel from \eqref{PB2}.
Noticing the elementary identity
\be
\langle \nabla_2 F - \nabla_2' F, \nabla_2' H + \nabla_2 H \rangle
= \langle \nabla_2 F, \nabla_2' H \rangle - \langle \nabla_2' F, \nabla_2 H \rangle,
\ee
we then obtain \eqref{PB2L} from \eqref{PB2}. The closure of the Poisson bracket
on the left $\U$ invariant functions follows from \eqref{PB2L}  and the transformation rules of the derivatives,
\bea
&&(\nabla_1 F) \circ \cL_\eta =\eta (\nabla_1 F) \eta^{-1},
\quad
(\nabla_1' F) \circ \cL_\eta =\nabla_1' F,
\nonumber\\
&& (\nabla_2 F) \circ \cL_\eta = \eta (\nabla_2 F) \eta^{-1},
\quad
(\nabla_2' F) \circ \cL_\eta = \eta (\nabla_2' F) \eta^{-1},
\eea
which hold for all $\eta \in \U$.
\end{proof}

We have seen that the $\U$ invariant functions form Poisson subalgebras with respect to the
second Poisson bracket \eqref{PB2}. Of course, the same is true regarding the first Poisson bracket \eqref{PB1}.
This can be seen directly from \eqref{PB1}, and also follows from well known  results about  cotangent
lifts of actions on a configuration space.

According to the singular value decomposition, also called Cartan (KAK) decomposition, every element $g\in \GL$ can be decomposed as
\be
g = \eta_L e^q \eta_R^{-1}, \quad
\eta_L, \eta_R \in \U, \quad q = \diag(q_1,q_2,\dots, q_n),
\quad
q_i\in \R,\,\, q_1 \geq q_2\geq \cdots \geq q_n.
\label{KAK}\ee
Here, $q$ is uniquely determined by $g$, and if
 $q_1 > q_2 >\cdots > q_n$ then $\eta_L$ and $\eta_R$ are unique up to the freedom
\be
(\eta_L, \eta_R) \mapsto (\eta_L \tau, \eta_R \tau), \quad \forall \tau \in \T^n,
\ee
with the maximal torus $\T^n < \U$.
In this paper, we call $\GL^\reg$ the dense open subset of $\GL$ whose elements obey the strict inequalities,
and we also introduce
\be
\fM^\reg := \GL^\reg \times \cG.
\label{fMreg}\ee
It is easily seen that every invariant function $F\in C^\infty(\fM)^{\U \times \U}$ can be recovered
from its restriction to the following submanifold of $\fM$:
\be
\fM_0^\reg:= \{ (e^q, J)\mid J\in \cG,\,\, q=\diag(q_1,q_2, \ldots , q_n),\, q_1 > q_2 >\cdots > q_n\}.
\label{fMreg0}\ee
With the tautological embedding $\iota: \fM_0^\reg \to \fM$, the restriction of $F$ reads
\be
f= F \circ \iota.
\ee
Obviously, $f$ is invariant under the $\T^n$ action on $\fM_0^\reg$ given by the
diffeomorphisms $A_\tau$,
\be
A_\tau: (e^q, J) \mapsto (e^q, \tau J \tau^{-1}),
\quad \forall \tau \in\T^n.
\label{Jtrans}\ee
That is to say, $f= F \circ \iota$ belongs to $C^\infty(\fM_0^\reg)^{\T^n}$.
Next, we introduce the reduced ring of functions,
\be
C^\infty(\fM)_\red:= \{ f\in C^\infty(\fM_0^\reg)^{\T^n}  \mid f = F \circ \iota,
\quad F \in C^\infty(\fM)^{\U\times \U}\}.
\label{redfun}\ee
This space of functions naturally inherits a pair of compatible Poisson brackets,
called \emph{reduced Poisson brackets}, which are
defined as follows:
\be
\{ F \circ \iota , H\circ \iota\}^\red_i := \{ F, H\}_i \circ \iota,
\quad\forall F,H \in C^\infty(\fM)^{\U \times \U}, \quad i=1,2.
\label{defredPB}\ee
The reduced Poisson brackets are well-defined, since the original Poisson brackets close on the
$\U \times \U$ invariant functions.
They contain all information about
 the
Poisson algebra carried by the (singular) quotient space
\be
\fM_\red:= \fM/ (\U\times \U),
\label{fMreddef}\ee
whose space of smooth functions is $C^\infty(\fM)^{\U \times \U}$.

Now our goal is to establish  intrinsic formulae of the reduced Poisson brackets, which contain only
derivatives with respect to the variables on $\fM_0^\reg$.
To this end, we need to express the derivatives of $F$ at $(e^q, J)\in \fM_0^\reg$ in terms of the derivatives
of the restricted function $f= F\circ \iota$.
We shall use the decompositions
\be
\cG^\pm  = \cG^\pm_0 + \cG^\pm_\perp,
\label{cGperp}\ee
where $\cG^\pm_0 \subset \cG^\pm$ contain the respective diagonal matrices, and $\cG^\pm_\perp$
contain the off-diagonal ones. Thus we can write $X^\pm = X_0^\pm + X^\pm_\perp$ for any $X\in \cG^\pm$.
Any function $f\in C^\infty(\fM_0^\reg)$ has the $\cG_0^-$-valued derivative $\nabla_1 f$ and the
$\cG$-valued derivative $d_2f$, determined by
\be
\langle  \nabla_1 f(e^q,J), X_0\rangle =  \dt f(e^{t X_0} e^q, J) =\dt f(e^{q+ t X_0}, J),
\quad
\forall X_0 \in \cG^-_0,
\label{der0}\ee
\be
\langle d_2 f(e^q,J),X\rangle = \dt f(e^q, J + t X),
\quad
\forall X\in \cG,
\ee
and we define $\nabla_2 f$ and $\nabla_2' f$ similarly to \eqref{E7}.
The definition \eqref{der0} makes sense since for small enough $t$ the components of $(q+ tX_0)$ satisfy
the same ordering condition as those of $q$.
Plainly, we have
\be
d_2 F(e^q,J) = d_2 f(e^q,J),
\ee
and as a result of the $\T^n$ invariance
\be
[J, d_2 f(e^q,J)]^+_0 =0.
\ee
Let us introduce the linear operator  $R(q)\in \mathrm{End}(\cG)$ by letting it act on an arbitrary matrix $X\in\cG$  according to
\be
(R(q) X)_{ii}=0,
\qquad
(R(q) X)_{ij} = X_{ij} \coth (q_i- q_j),
\qquad
 1\leq i\neq j \leq n.
 \label{Rq}\ee
 Notice that $R(q)$ maps $\cG^\pm_\perp$ onto $\cG^\mp_\perp$ \eqref{cGperp}, respectively, in an invertible manner, and it satisfies
 \be
 \langle R(q) X, Y \rangle = - \langle X, R(q) Y \rangle,
 \quad
 \forall X,Y \in \cG.
 \label{Rasym}\ee
 Due to the following result, $R(q)$  will appear in the expressions of the reduced Poisson brackets.

  \medskip\noindent
\begin{lemma} \label{Lem:3.3}
If $F\in C^\infty(\fM)^{\U \times \U}$ and $f=F\circ \iota$, then we have
\be
\nabla_1 F(e^q, J) = [d_2 f(e^q,J), J]^+ + \nabla_1 f(e^q,J) + R(q) [d_2 f(e^q,J), J]^+,
\label{nabF}\ee
where the superscripts refer to the decomposition \eqref{Gpm}.
\end{lemma}
\begin{proof}
We know from \eqref{rightinv} that $(\nabla_1'F)^+ =0$ holds, because $F\in C^\infty(\fM)^{\U}_\cR$.
On the other hand, $(\nabla_1F)^+ = (\nabla_2' F - \nabla_2 F)^+$ holds \eqref{id1},
because $F\in C^\infty(\fM)^{\U}_\cL$.
By using that
\be
\nabla_1' F(e^q, J) = e^{-q} (\nabla_1 F(e^q,J)) e^q,
\ee
we can write
\be
0 = (\nabla_1'F )^+ = (\cosh \ad_q)( (\nabla_1F)^+) - (\sinh \ad_q) ((\nabla_1 F)^-),
\ee
at any $(e^q,J)$, where $\ad_q(X):= [q,X]$.
This implies that
\be
(\nabla_1 F)^-_\perp = R(q) (\nabla_1 F)^+.
\ee
Finally, since
\be
\dt F(e^{t X_0} e^q, J) = \dt f( e^{tX_0} e^q, J),
\quad
\forall X_0\in \cG^-_0,
\ee
we get
\be
(\nabla_1 F(e^q, J))^-_0 = \nabla_1 f(e^q, J).
\ee
The proof is completed by noting that
\be
\nabla_2' F - \nabla_2 F = \nabla_2' f - \nabla_2 f  = [d_2 f, J] \quad \hbox{at any}\quad (e^q,J)\in \fM_0^\red.
\ee
\end{proof}

 \medskip\noindent
\begin{theorem} \label{Prop:3.4}
For $f,h \in C^\infty(\fM)_\red$ \eqref{redfun}, the first reduced Poisson bracket \eqref{defredPB} is given by
\bea
\{ f,h\}_1^\red(e^q,J) &=&  \langle \nabla_1 f, (d_2h)^-_0\rangle - \langle \nabla_1 h, (d_2 f)^-_0 \rangle \label{redPB1}\\
&+&\langle R(q) [d_2 f, J]^+, (d_2 h)^- \rangle   - \langle R(q) [d_2 h, J]^+, (d_2 f)^- \rangle
\nonumber \\
 &+& \langle J^+, [ (d_2 f)^-, (d_2 h)^-] - [ (d_2 f)^+, (d_2 h)^+] \rangle,\nonumber
\eea
where the derivatives are taken at $(e^q,J)$ and we use $R(q)$ \eqref{Rq}.
If $h$ is the restriction of any of the Hamiltonians $H_k$ or $\tilde H_k$ \eqref{HamsonfM}, then
this formula can be written as
\be
\{f,h\}_1^\red(e^q,J) = \langle \nabla_1 f, (d_2 h)^-_0 \rangle +
\langle d_2 f,
[R(q) (d_2 h)^- - (d_2 h)^+, J] \rangle.
\label{redham1}\ee
 \end{theorem}
\begin{proof}
The first and second lines of \eqref{redPB1} represent the contributions of the second and third terms of
\eqref{nabF} obtained  upon substitution in the formula \eqref{PB1}.
Regarding the third line, it arises from the first term of \eqref{nabF} by taking the sum of
\be
\langle [d_2 f, J]^+ , d_2 h\rangle  -  \langle [d_2 h, J]^+ , d_2 f\rangle  =
- 2 \langle J^+,  [(d_2 f)^+, (d_2 h)^+] \rangle
- \langle J^-, [ d_2 f, d_2 h]^- \rangle
\ee
and
\be
\langle J, [d_2 f, d_2 h]\rangle = \langle J^+, [d_2 f, d_2 f]^+ \rangle + \langle J^-, [d_2 f, d_2 h]^-\rangle.
\ee
To get \eqref{redham1} from \eqref{redPB1}, we use that for any Hamiltonian $h$ proportional with the
real or imaginary parts of $\tr(J^k)$ one has $[d_2 h, J ] =0$.
In this case, the second line of \eqref{redPB1} gives
\be
\langle R(q) [d_2 f, J]^+, (d_2 h)^- \rangle =\langle R(q) [d_2 f, J], (d_2 h)^- \rangle =
\langle d_2 f, [R(q) (d_2 h)^-, J]\rangle,
\ee
and the third line becomes
\bea
&&\langle J^+, [ (d_2 f)^-, (d_2 h)^-] - [ (d_2 f)^+, (d_2 h)^+] \rangle =
\langle J, [ (d_2 f)^-, (d_2 h)^-] - [ (d_2 f)^+, (d_2 h)^+] \rangle \nonumber\\
&& \qquad =\langle J, [  (d_2 h)^+, (d_2 f)^-] + [  (d_2 h)^+, (d_2 f)^+] \rangle
= -\langle d_2 f, [ (d_2 h)^+, J] \rangle,
\eea
which confirm the claim \eqref{redham1}.
\end{proof}

 \medskip\noindent
\begin{theorem} \label{Prop:3.5}
For $f,h \in C^\infty(\fM)_\red$ \eqref{redfun}, the second reduced Poisson bracket \eqref{defredPB}  can be written as
\bea
&& 2\{f,h\}_2^\red(e^q, J) =
\langle \nabla_1 f, (\nabla _2h + \nabla'_2 h)^-_0\rangle - \langle \nabla_1 h, (\nabla_2 f + \nabla_2' f)^-_0
\rangle \label{redPB2}\\
&&\qquad +\langle R(q) [d_2 f, J]^+, (\nabla_2 h +\nabla_2' h)^- \rangle   -
\langle R(q) [d_2 h, J]^+, (\nabla_2 f +\nabla_2' f)^- \rangle
\nonumber \\
&&\qquad + \langle (\nabla_2 f)^-, (\nabla_2' h)^-\rangle  +\langle (\nabla_2' f)^+, (\nabla_2 h)^+\rangle
 -(\nabla_2' f)^-, (\nabla_2 h)^-\rangle  -\langle (\nabla_2 f)^+, (\nabla_2' h)^+\rangle,
 \nonumber
\eea
where the derivatives are taken at $(e^q,J)$ and $R(q)$ is given by \eqref{Rq}.
If $h$ is the restriction of any of the Hamiltonians $H_k$ or $\tilde H_k$ \eqref{HamsonfM}, then this formula
can be recast in the form
\be
\{f,h\}_1^\red(e^q,J) = \langle \nabla_1 f, (\nabla_2 h)^-_0 \rangle +
\langle d_2 f, [R(q) (\nabla _2 h)^- - (\nabla_2 h)^+, J] \rangle.
\label{redham2}\ee
\end{theorem}
\begin{proof}
We see by combining \eqref{310} (obtained in the proof of Lemma \ref{Lem:3.1}) and Lemma \ref{Lem:3.2}
that for $F,H\in C^\infty(\fM)^{\U \times \U}$
\be
2\{ F, H\}_2  =
 \langle \nabla_2 F ,\nabla_2' H \rangle
 - \langle \nabla_2 H ,\nabla_2' F \rangle
 + \langle \nabla_1 F,  \nabla_2' H  +\nabla_2 H  \rangle
-  \langle \nabla_1 H,   \nabla_2' F + \nabla_2 F \rangle.
\label{PB2LR}\ee
Putting $f= F\circ \iota$ and $h= H\circ \iota$, we substitute the identity \eqref{nabF}
both for $\nabla_1 F$ and $\nabla_1H$. After that, we spell out the first two terms of \eqref{PB2LR}
and the contributions
coming from the first terms of $\nabla_1 F$ and $\nabla_1 H$.
The formula  \eqref{redPB2} is then obtained by collecting terms.
Turning to the proof of \eqref{redham2}, instead of \eqref{redPB2} it is shorter to go back to
\eqref{PB2LR}.  We observe from \eqref{d2Hams}  that in this case
$[d_2 H, J]=0$. Thus $\nabla_2 H = \nabla_2' H$ and
\be
\langle \nabla_2 F ,\nabla_2' H \rangle
 - \langle \nabla_2 H ,\nabla_2' F \rangle =\langle d_2 F, [d_2H, J^2] \rangle = 0.
 \ee
Then
\bea
\{f,h\}_2^\red(e^q,J)&=& \langle \nabla_1 F (e^q, J), \nabla_2 H(e^q, J) \rangle  \\
&=& \langle [d_2 f(e^q,J), J]^+ + \nabla_1 f(e^q,J) + R(q) [d_2 f(e^q,J), J]^+, \nabla_2 h(e^q, J)\rangle.
\nonumber
\eea
The verification is finished by noting that
\be
\langle R(q) [d_2 f, J]^+, \nabla_2 h\rangle = - \langle [d_2 f, J], R(q) (\nabla_2 h)^- \rangle=
\langle d_2 f,  [R(q) (\nabla_2 h)^-, J] \rangle,
\ee
and
\be
\langle [d_2 f, J]^+, \nabla_2 h\rangle = \langle d_2f, [J, (\nabla_2 h)^+]\rangle .
\ee
\end{proof}

The following statement summarizes the outcome of our construction.
\medskip\noindent
\begin{proposition}  \label{Th:3.6}
 The formulae \eqref{redPB1} and \eqref{redPB2} yield two compatible Poisson brackets
 on $C^\infty(\fM)_\red$ \eqref{redfun}. The commuting Hamiltonians
 \be
 h_k:= \frac{1}{k} \Re \tr(J^k)
 \quad\hbox{and}\quad
 \tilde h_k := \frac{1}{k} \Im \tr(J^k),
 \qquad
 k\in \N,
 \label{redhams}\ee
 give rise to bi-Hamiltonian evolution equations since they satisfy
 \be
 \{ f, h_k\}_2^\red = \{ f, h_{k+1}\}_1^\red
 \quad \hbox{and}\quad
 \{ f, \tilde h_k\}_2^\red = \{ f, \tilde h_{k+1}\}_1^\red,
 \quad
 \forall f\in C^\infty(\fM)_\red,\,\, k\in \N.
 \label{redbih}\ee
 \end{proposition}
\begin{proof}
This is obvious from our construction.
In particular,
the compatibility of the reduced Poisson brackets follows from the compatibility
of the original Poisson brackets on $C^\infty(\fM, \R)$ by applying the definition \eqref{defredPB}.
The properties \eqref{redbih} are consequences of \eqref{redham1} and \eqref{redham2} taking into account that
\be
\nabla_2 h_k = J^k= d_2 h_{k+1}
\quad
\hbox{and}
\quad
\nabla_2 \tilde h_k = -\ri J^k = d_2 \tilde h_{k+1},
\label{dershk}\ee
which are implied by \eqref{d2Hams}.
\end{proof}

In the next remark, we explain that the compatible Poisson brackets can be defined also on
the ring of function $C^\infty(\fM^\reg_0)^{\T^n}$.

 \medskip\noindent
\begin{remark} \label{Rem:3.7}
The ring of functions $C^\infty(\fM)_\red$ \eqref{redfun} is contained in $C^\infty(\fM_0^\reg)^{\T^n}$,
but is not equal to it.
For example, the components of $q$ give elements of $C^\infty(\fM_0^\reg)^{\T^n}$, and can be extended to unique,
$\U \times \U$ invariant continuous  functions on $\fM$,
 but these
functions lose their differentiability at the locus where $q_i= q_{i+1}$ for some $i$.
This holds since the $e^{2 q_i}$ \eqref{KAK} are the ordered eigenvalues of $g g^\dagger$ and, as is well known,  the differentiability
of eigenvalues
is in general lost where they coincide.
Nevertheless, the formulae \eqref{redPB1} and \eqref{redPB2} define compatible Poisson brackets
on the whole of $C^\infty(\fM_0^\reg)^{\T^n}$. In order to see this, consider $\fM^\red$ \eqref{fMreg},
which is the set of elements
of $\fM$ that can be transformed into $\fM_0^\reg$ by the action of $\U \times \U$.
This is a dense open subset and the pull-back by $\iota: \fM_0^\reg \to \fM^\reg$  yields an injective
and \emph{surjective} map
from $C^\infty(\fM^\reg)^{\U \times \U}$ onto $C^\infty(\fM_0^\reg)^{\T^n}$, that is,
\be
C^\infty(\fM_0^\reg)^{\T^n}  = \{ f \in C^\infty(\fM_0^\reg) \mid f = F \circ \iota,
\quad F \in C^\infty(\fM^\reg)^{\U\times \U}\}.
\label{redfun0}\ee
Then the application of \eqref{defredPB} to $F, H \in C^\infty(\fM^\reg)^{\U \times \U}$
gives rise to compatible Poisson brackets on $C^\infty(\fM_0^\reg)^{\T^n}$.
They are described by the formulae \eqref{redPB1} and \eqref{redPB2} for any
$f,h\in C^\infty(\fM_0^\reg)^{\T^n}$,
since they are determined by calculations identical to those presented above.
\end{remark}

For a vector field $\cE$ on $\fM_0^\reg$, we denote the derivative of $f\in C^\infty(\fM_0^\reg)$ by
$\cE[f]$. The vector field $\cE$ is encoded by the matrix valued functions $\cE[q]$ and $\cE[J]$, i.e.,
by the derivatives of $q$ and $J$ regarded as evaluation functions that return $q$ and $J$ when
applied to $(q,J) \in \fM_0^\reg$.
Then the chain rule reads
\be
\cE[f] = \langle \nabla_1 f, \cE[q] \rangle +  \langle d_2 f, \cE[J] \rangle.
\ee
For any fixed $h \in C^\infty(\fM_0^\reg)^{\T^n}$, the two Poisson brackets with $h$ determine
two derivations of  $C^\infty(\fM_0^\reg)^{\T^n}$.  These correspond to vector fields $\cE_h^i$ ($i=1,2$) that are unique only
up to the addition of infinitesimal gauge transformations.
The term \emph{infinitesimal gauge transformation} refers to any vector field $\cZ$ for which
\be
\cZ[q]=0
\quad\hbox{and}\quad
\cZ[J] = [\cT, J]
\label{cZ}\ee
with some function $\cT: \fM_0^\reg\to \cG^+_0$.
Note that $\cG_0^+$ is just the Lie algebra of $\T^n$.
This ambiguity drops out after projection to the quotient space $\fM_0^\reg/\T^n$.
For definiteness, we shall fix this ambiguity of the vector field $\cE_h^i$ by
 by imposing the condition
 \be
\cE_h^i[f] = \{f,h\}_i^\red,
\qquad \forall f\in C^\infty(\fM_0^\reg),
\ee
 where $\{f, h\}_i^\red$ is understood to be given by the formulae  \eqref{redPB1} and \eqref{redPB2}.
These formulae define anti-symmetric bi-derivations on $C^\infty(\fM_0^\reg)$, but the Jacobi
identity holds only for the $\T^n$ invariant functions.
By some abuse of terminology, we call the vector field $\cE^i_h$  \emph{the Hamiltonian vector field}
associated with $h$ by means of the bracket $\{\ ,\ \}_i^\red$.

\medskip\noindent
\begin{proposition}\label{Prop:3.8}
Consider the vector fields  $\cE_k^i$ and $\tilde \cE_k^i$ defined by
\be
\cE_k^i[f] = \{ f, h_k\}_i^\red
\quad
\hbox{and}\quad
\tilde \cE_k^i[f] = \{ f, \tilde h_k\}_i^\red, \quad k\in \N,\, i=1,2,
\ee
with the 
Hamiltonians \eqref{redhams}. Using $R(q)$ \eqref{Rq}, these vector fields have the explicit form
\be
\cE_k^2[q_j]= \cE_{k+1}^1[q_j] = \Re (J^k)_{jj},
\quad
\cE_k^2[J] = \cE_{k+1}^1 [J] = \frac{1}{2}[ R(q)(J^k + (J^k)^*)  + ((J^k)^* - J^k ), J],
\label{cE}\ee
and
\be
\tilde \cE_k^2[q_j]= \tilde \cE_{k+1}^1[q_j] = \Re (-\ri J^k)_{jj},
\quad
\tilde \cE_k^2[J] = \tilde \cE_{k+1}^1 [J] =
\frac{1}{2}[\ri R(q)(  (J^k)^* - J^k   )  + \ri ( J^k + (J^k)^*), J].
\label{tcE}\ee
All these vector fields are tangent to the submanifold
\be
\fM_{0,-}^\reg = \{ (q,J^-) \in \fM_0^\reg \mid J^- \in \cG^-\} \subset \fM_0^\reg
\label{fM-}\ee
defined by imposing the constraint $J^+=0$, as well as to the submanifold
\be
\fM_{0,+}^\reg = \{ (q,J^+) \in \fM_0^\reg \mid J^+ \in \cG^+\} \subset \fM_0^\reg
\label{fM+}\ee
defined by imposing the constraint $J^-=0$.
The restriction $\cV_k^i$ of $\cE_k^i$  on $\fM_{0,-}^\reg$ gives
\be
\cV_k^2[q_j] = \cV_{k+1}^1 [q_j] = (J^-)^k_{jj},
\quad
\cV_k^2[J^-] = \cV_{k+1}^1 [J^-] = [ R(q)(J^-)^k, J^-],
\label{cV}\ee
while the restriction
of $\tilde \cE_k^i$ vanishes identically on $\fM_{0,-}^\reg$  for all $k$ and $i$.
The vector fields $\cV_k^i$ reproduce the evolutional vector fields of the
spin Sutherland
hierarchy described in \cite{FeNon}.

Denoting the restrictions of $\cE_k^i$ and $\tilde \cE_k^i$ on $\fM_{0,+}^\reg$  by $\cU_k^i$ and $\tilde \cU_k^i$, we obtain
\be
\cU^2_{2l-1} = \cU^1_{2l} =0,
\quad
\tilde \cU_{2l}^2 = \tilde \cU_{2l+1}^1 =0,
\ee
\bea
&&\cU_{2l}^2[q_j] = \cU_{2l+1}^1[q_j] = \Re((J^+)^{2l})_{jj},
\quad
\cU_{2l}^2[J^+] = \cU_{2l+1}^1[J^+] = [ R(q) (J^+)^{2l}, J^+],
\label{cU}\\
&&\tilde \cU_{2l-1}^2[q_j] = \tilde \cU_{2l}^1[q_j] = \Re (-\ri  (J^+)^{2l-1})_{jj},
\quad
\tilde \cU_{2l-1}^2[J^+] = \tilde \cU_{2l}^1[J^+] = [ -\ri R(q) (J^+)^{2l-1}, J^+]\nonumber.
\eea
By making the substitution $J^+ = \ri J^-$, the vector fields shown in \eqref{cU} get
transformed into those that appear in \eqref{cV}, up some irrelevant overall signs.
\end{proposition}

\begin{proof}
The formulae \eqref{cE} and \eqref{tcE} follow directly by applying the formulae \eqref{redham1}, \eqref{redham2}
and \eqref{dershk}.
 The tangency
 to $\fM^\reg_{0,\pm}$ is a consequence of the fact that
$[\cG^+,\cG^\pm ]\subset \cG^{\pm}$.
The statements about the form  of the restricted vector fields are plain from
 \eqref{cE} and \eqref{tcE}.
 Comparison with  equation (1.8) in \cite{FeNon} shows that
 the bi-Hamiltonian vector fields \eqref{cV} reproduce the spin Sutherland hierarchy
studied earlier.
\end{proof}

 \medskip\noindent
\begin{remark} \label{Rem:3.9}
It is worth noting that the vector fields $\cZ$ \eqref{cZ} representing infinitesimal gauge transformations
are tangent to the submanifolds $\fM_{0,\pm}^\reg$. Therefore the possibility to restrict a Hamiltonian vector field  from
$\fM_0^\reg$  to $\fM_{0,\pm}^\reg$ is independent of the ambiguity of adding an infinitesimal gauge transformation.
\end{remark}

\section{Interpretation as spin Sutherland models}
\setcounter{equation}{0}

We below develop the physical interpretation of the reduced bi-Hamiltonian system.
This interpretation will be reached via a suitable
parametrization of the variable $J$ for $(e^q, J) \in \fM_0^\reg$. In this parametrization
the first reduced Poisson structure and one of the reduced Hamiltonians take the form
characteristic of spin Sutherland models.
 We first present the particular case corresponding to
the Hamiltonian \eqref{spinSuth1} and deal
with the general case  \eqref{spinSuth2} subsequently.
We will not spell out the expression of the second
Poisson bracket in terms of the new variables, since
those formulae are complicated and do not enhance our
understanding.

\subsection{Spin Sutherland model on the subspace $\fM_{0,-}^\reg \subset \fM_0^\reg$}

We have seen in Proposition \ref{Prop:3.8}  that the evolutional vector fields of the spin Sutherland hierarchy
(given by \eqref{cV})
result from  our construction by restriction to the submanifold $\fM_{0,-}^\reg$.
Now we show that the bi-Hamiltonian structure found in \cite{FeNon}
also results from this restriction.  We start by demonstrating that
the derivative of $J^+$ vanishes along all Hamiltonian vector fields $\cE_h^i$
at those point where $J^+=0$, i.e., on the submanifold $\fM^\reg_{0,-}$ \eqref{fM-}.
To see this, we rearrange the Poisson brackets in the form
\be
\{f,h\}_i^\red = \langle \nabla_1 f, \cE_h^i[q] \rangle + \langle (d_2f)^-, \cE_h^i[J^-]\rangle +
\langle (d_2 f)^+, \cE_h^i[J^+] \rangle,
\ee
and show that the functions
\be
\cE_h^i[J^+]: \fM_0^\reg \to \cG^+
\ee
vanish for both Poisson brackets upon imposing the constraint $J^+=0$.

The following lemmas are obtained by straightforward,
somewhat tedious, calculations. We sketch only the proof of the second lemma, which is the more
complicated one.

\begin{lemma}\label{lem:4.1}
The Hamiltonian vector field $\cE_h^1$ associated with the Poisson bracket \eqref{redPB1} has the components
\be
\cE_h^1[q] = (d_2h)^-_0,
\ee
\be
\cE_h^1[J^-] = -\nabla_1 h  + [(d_2h)^-, J^+] + R(q)[J, d_2h]^+ + [R(q) (d_2h)^-, J^-],
\ee
\be
\cE_h^1[J^+]= [ R(q) (d_2h)^- - (d_2 h)^+, J^+].
\ee
Consequently, $\cE_h^1$ is tangent to the submanifold $\fM^\reg_{0,-}$ \eqref{fM-}
for every $h\in C^\infty(\fM_0^\reg)^{\T^n}$.
 \end{lemma}

\begin{lemma}\label{lem:4.2}
The Hamiltonian vector field $\cE_h^2$ associated with the Poisson bracket \eqref{redPB2} has the components
\be
\cE_h^2[q] = \left(J^- (d_2h)^- + J^+ (d_2 h)^+\right)^-_0,
\label{L1}\ee
\be
\cE_h^2[J^-] =  \left(2J^- (d_2h)^- J^+  - 2J^- R(q)((d_2h)^- J^- + (d_2h)^+ J^+) - \nabla_1 h J^-\right)^-,
\label{L2}\ee
\be
\cE_h^2[J^+]= \left( 2 J^+ (d_2h)^+ J^-   - 2 J^+ R(q)( (d_2h)^- J^- + (d_2h)^+ J^+)  - \nabla_1 h J^+ \right)^+.
\label{L3}\ee
It follows that $\cE_h^2$ is tangent to  both submanifolds $\fM^\reg_{0,-}$ \eqref{fM-} and
$\fM^\reg_{0,+}$ \eqref{fM+}.
 \end{lemma}
 \begin{proof}
 We shall use that, for any $X\in \cG$, $X^- = \frac{1}{2} (X + X^*)$ and $X^+ = \frac{1}{2} (X- X^*)$, together
 with obvious properties of the trace form \eqref{E1}, like $\langle X^*, Y^*\rangle = \langle X, Y\rangle$.
 By directly spelling it out, we find
 \be
 (\nabla_2 h + \nabla_2' h)^- = 2 \left( J^- (d_2 h)^- + J^+ (d_2 h)^+ \right)^-.
 \ee
 Thus, the first term of $\{f,h\}_2^\red$ \eqref{redPB2} gives \eqref{L1}, and the second term of \eqref{redPB2} gives the last terms in both lines
 \eqref{L2} and \eqref{L3}. In order to confirm the latter statement, note that
 \bea
 &&\frac{1}{2} \langle \nabla_1 h, (\nabla_2 f + \nabla_2' f)^-_0 \rangle = \langle \nabla_1 h, J^- (d_2 f)^- + J^+ (d_2 f)^+ \rangle =
 \\
 &&
   \langle \nabla_1 h J^-, (d_2 f)^- \rangle  + \langle \nabla_1 h J^+, (d_2 f)^+ \rangle
  = \langle (\nabla_1 h J^-)^-, (d_2 f)^- \rangle  + \langle (\nabla_1 h J^+)^+, (d_2 f)^+ \rangle. \nonumber
 \eea
 Next, we inspect the terms of \eqref{redPB2} that contain $R(q)$ \eqref{Rq}. We use the anti-symmetry of $R(q)$ \eqref{Rasym}, the property
 $(R(q) X)^* = - R(q) X^*$, and that it maps $\cG^\pm$ into $\cG^{\mp}$, respectively.  Then we can write
 \bea
&&\langle R(q) [d_2f, J]^+, (\nabla_2 h + \nabla_2' h)^- \rangle  \nonumber\\
&&\,\, = 2 \langle R(q)\left( (d_2 f)^+ J^+ - J^+ (d_2 f)^+ + (d_2 f)^- J^- - J^- (d_2 f)^- \right), J^- (d_2 h)^- + J^+ (d_2 h)^+ \rangle \nonumber \\
&&\,\, =2 \langle R(q)\left( (d_2 f)^+ J^+  + (d_2 f)^- J^-  \right), J^- (d_2 h)^- +    J^+ (d_2 h)^+  + (d_2 h)^- J^- + (d_2 h)^+ J^+ \rangle
\nonumber \\
&&\,\,=  - 2  \langle (d_2 f)^+ , J^+ R(q) \left( J^- (d_2 h)^- +    J^+ (d_2 h)^+  + (d_2 h)^- J^- + (d_2 h)^+ J^+ \right) \rangle \nonumber\\
&&\,\quad  - 2  \langle (d_2 f)^- , J^- R(q) \left( J^- (d_2 h)^- +    J^+ (d_2 h)^+  + (d_2 h)^- J^- + (d_2 h)^+ J^+ \right) \rangle.
 \eea
Similarly, we obtain
\bea
&& - \langle R(q) [d_2 h, J]^+, (\nabla_2 f + \nabla_2' f)^- \rangle
= - 2 \langle R(q) [d_2 h, J]^+, (d_2 f)^- J^- +(d_2 f)^+  J^+  \rangle  \nonumber\\
&&\qquad   = - 2 \langle (d_2 f)^+,  J^+ R(q) [d_2 h, J]^+ \rangle - 2  \langle (d_2 f)^-,  J^- R(q) [d_2 h, J]^+ \rangle.
\eea
By adding these two expressions, and taking into account the factor $2$ on the left-hand side of \eqref{redPB2},
we get the terms containing $R(q)$ in \eqref{L2} and \eqref{L3}.
The first terms on the right-hand sides of \eqref{L2} and \eqref{L3} result  by expanding and collecting
all terms coming from the last line of
\eqref{redPB2}, which is laborious but fully straightforward.
Some further details are given in Appendix \ref{app:details}.
 \end{proof}

Let us remember \cite{OR} that
a Poisson submanifold $N$ of a Poisson manifold $(M, \{\ , \}_M)$ is characterized by the property
that if one considers any Hamiltonian vector field on $M$ and restricts it to $N$, then the
restricted vector field is tangent to $N$. Under this condition, one obtains a Poisson structure
$\{\ ,\ \}_N$ on $N$ as follows.  Take any smooth functions $\cF, \cH$ on $N$ and extend them arbitrarily
to smooth functions $f, h$ on $M$. (It is sufficient to consider such extensions only locally, and $N$ can be
an immersed submanifold).
Then the formula
\be
\{\cF, \cH\}_N(x) := \{ f, h\}_M(x),
\qquad
\forall x\in N,
\ee
gives a well-defined Poisson bracket on $N$.
It follows from Lemma \ref{lem:4.1} and Lemma \ref{lem:4.2} that this procedure can be applied in our situation, too, and thus
we obtain well-defined Poisson brackets on $C^\infty(\fM^\reg_{0,-})^{\T^n}$ by restriction
of the Poisson brackets on $C^\infty(\fM_0^\reg)^{\T^n}$.
The variables of a function $\cF \in C^\infty(\fM^\reg_{0,-})$ are given by the pair $(e^q, J^-)$.
 Mimicking the previous practice, we introduce the corresponding $\cG^-_0$-valued and
$\cG^-$-valued derivatives $\nabla_1 \cF$ and $d_2 \cF$, and also set
$\nabla_2 \cF := J^- d_2 \cF$ and $\nabla_2' \cF :=  d_2 \cF J^-$.

\begin{theorem}\label{Th:4.3}
The compatible Poisson brackets given by \eqref{redPB1} and \eqref{redPB2} on $C^\infty(\fM^\reg_0)^{\T^n}$
can be restricted to $C^\infty(\fM^\reg_{0,-})^{\T^n}$. For $\cF, \cH \in C^\infty(\fM^\reg_{0,-})^{\T^n}$
the restricted brackets, denoted  $\{\cF ,\cH \}_{i,-}^{\red}$, can be written as
\be
\{\cF, \cH\}_{1,-}^{\red}(e^q, J^-) = \langle \nabla_1 \cF, d_2 \cH \rangle - \langle \nabla_1 \cH, d_2 \cF\rangle
+ \langle J^-, [R(q) d_2\cF, d_2\cH ]  + [ d_2 \cF, R(q) d_2 \cH] \rangle
\label{-PB1}\ee
and
\be
\{ \cF, \cH\}^{\red}_{2, -}(e^q, J^-) = \langle \nabla_1 \cF, \nabla_2 \cH \rangle
- \langle \nabla_1 \cH, \nabla_2 \cF \rangle + 2 \langle \nabla_2 \cF, R(q) \nabla_2 \cH\rangle,
\label{-PB2}\ee
where the  derivatives are evaluated at $(e^q, J^-)$.
These formulae  reproduce the compatible
Poisson brackets of the hyperbolic spin Sutherland hierarchy
described earlier in \cite{FeNon}.
\end{theorem}

\begin{proof}  Suppose that $\cF, \cH$ are the restrictions of $f, h \in C^\infty(\fM^\reg_0)^{\T^n}$.
Then, according to the definition of the restricted brackets,
\be
\{\cF, \cH\}_{i,-}^{\red}(e^q, J^-)
= \left(\langle \nabla_1 f, \cE_h^i[q] \rangle + \langle (d_2f)^-, \cE_h^i[J^-]\rangle\right)(e^q, J^-),
\ee
where the formulae of the preceding two lemmas are applied, at $J^+=0$.
Here, we substitute $(d_2 f)^-(e^q, J^-) = d_2 \cF(e^q, J^-)$ and analogous relations
for the other derivatives.
 This readily leads to the above forms of the restricted Poisson brackets.
For example, to obtain \eqref{-PB2}, we also use that
\be
-\langle d_2 \cF, 2J^- R(q)\nabla'_2 \cH \rangle = - 2 \langle \nabla'_2 \cF, R(q) \nabla'_2 \cH
\rangle = 2  \langle \nabla_2 \cF, R(q) \nabla_2 \cH \rangle.
\ee
This holds by virtue of the identities
\be
\nabla_2 \cH = (\nabla_2' \cH)^*, \quad R(q) X^* = - (R(q) X)^*,\quad
\langle X^*, Y^*\rangle = \langle X, Y\rangle,\,\,  \forall X,Y\in \cG.
\ee
Taking into account some obvious differences of notation,  one sees by direct comparison
that the Poisson brackets in \eqref{-PB1} and \eqref{-PB2}  coincide with those in Theorem 1
of \cite{FeNon}.
\end{proof}

\begin{remark}\label{Rem:4.4}  The Poisson bracket \eqref{-PB2} was obtained in \cite{FeNon} by
suitably rewriting the Poisson structure of an example of models
derived by Li \cite{Li} applying a rather complicated method based on dynamical Poisson--Lie groupoids.
Then it was directly shown to be compatible with the first Poisson bracket \eqref{-PB1}
extracted from \cite{FP1}.
The present derivation is simpler and it highlights that both Poisson brackets originate
from a single reduction in a unified manner.
If we parametrize the Hermitian matrix $J=J^-$ in the form
\be
(J^-)_{ij} = p_i \delta_{ij} - (1- \delta_{ij}) \frac{\xi_{ij}}{\sinh(q_i - q_j)},
\label{J-par}\ee
where the $p_i$ are arbitrary real numbers and $\xi$ in an off-diagonal anti-Hermitian matrix,
then the reduced Hamiltonian $\frac{1}{2} \tr(J^2)$ reproduces \eqref{spinSuth1}.
The spin $\xi$ matters up to the gauge transformations $\xi \mapsto \tau \xi \tau^{-1}$, $\forall \tau \in \T^n$.
Under the reduced first Poisson bracket \eqref{-PB1}, the $\T^n$ invariant functions of $\xi$
are those arising from the $\u(n)$ Lie-Poisson bracket reduced
by the first class constraints $\xi_{kk} =0$ for all $k$.
The $q_i, p_i$ $(i=1,\dots, n)$ form canonical pairs with respect to  the reduced first Poisson bracket,
and they Poisson commute with the $\T^n$ invariant functions of $\xi$.
These statements are proved in \cite{FeNon,FP1}, and will be generalized in the next subsection.
\end{remark}

The combined message of Theorem \ref{Th:4.3} and Proposition \ref{Prop:3.8}  is summarized in the next corollary.

\begin{corollary}\label{Cor:4.5}
The restrictions of the Hamiltonians $h_k$ \eqref{redhams} on $\fM^\reg_{0,-}$,
given by $\cH_k(e^q, J^-) = \frac{1}{k} \tr( (J^-)^k)$, together with
the compatible Poisson brackets of Theorem \ref{Th:4.3} reproduce the bi-Hamiltonian vector fields
\eqref{cV} of the spin Sutherland hierarchy.  By using the parametrization \eqref{J-par},
$\cH_2(e^q, J^-)$  turns into\footnote{This justifies calling the system `spin Sutherland hierarchy'.}
 $\cH_{\mathrm{spin}-1}$ \eqref{spinSuth1}.
 The Hamiltonians $\tilde h_k$ \eqref{redhams} vanish
 on $\fM^\reg_{0,-}$.
\end{corollary}

Finally, let us observe that $J^+$ and $J^-$ appear rather symmetrically in the formula of Lemma \ref{lem:4.2}. In particular,
$\cE_h^2[J^-] = 0$ holds
after restriction to  the submanifold $\fM^\reg_{0,+}$ \eqref{fM+}.
Therefore we can restrict the second reduced Poisson bracket on this submanifold.
Moreover, $n$ out of the $2n$ commuting Hamiltonians \eqref{redhams} survives this restriction,
in correspondence to the vector fields in \eqref{cU}.
It can be verified that the substitution $J^+ = \ri J^-$, with the new variable $J^-$,
converts the restricted Poisson bracket
on $C^\infty(\fM^\reg_{0,+})^{\T^n}$ into the restricted second Poisson bracket on
$C^\infty(\fM^\reg_{0,-})^{\T^n}$.
This means that we do not obtain anything new from this restriction, and hence we omit
its more detailed description.

\subsection{The general case: Sutherland model with two spins}

We now explain how the generalized spin Sutherland Hamiltonian \eqref{spinSuth2} arises from our reduced system.
For this purpose, we take $(e^q, J)\in \fM_0^\reg$ and (applying \eqref{Gpm}) define the new variables
\be
\xi^l:= - J^+, \quad
\xi^r:= (e^{-q} J e^q)^+,
\quad
p_k:= J^-_{kk},\,\, k=1,\dots, n.
\label{xilrdef}\ee
We observe that the pair $(\xi^l, \xi^r)$ obeys the constraints
\be
\xi^l_{kk} + \xi^r_{kk} =0,
\qquad
k=1,\dots, n.
\label{xilr0}\ee

\begin{lemma}\label{lem:4.6}
The matrix $J$ can be reconstructed from $(q,p, \xi^l, \xi^r)$ defined by \eqref{xilrdef}
according to the formula
\be
J_{ij} = p_i \delta_{ij} - (1 - \delta_{ij}) \left( \coth(q_i -q_j) \xi^l_{ij} + \xi^r_{ij}/\sinh(q_i-q_j)\right) - \xi^l_{ij},\qquad\forall
1\leq i,j\leq n.
\label{Jpar}\ee
This expression provides a parametrization of $\fM_0^\reg$ \eqref{fMreg0}  by the variables $(q,p, \xi^l, \xi^r)$, where  $q_1 > q_2 >\dots > q_n$, the
$p_k \in \R$ are arbitrary and $(\xi^l, \xi^r) \in \u(n) \oplus \u(n)$ is subject to the constraints \eqref{xilr0}.
The residual gauge transformations act on $(q,p, \xi^l, \xi^r)$   according to
\be
(q,p, \xi^l, \xi^r) \mapsto (q,p, \tau \xi^l \tau^{-1}, \tau \xi^r \tau^{-1}),
\quad
\forall
\tau \in \T^n.
\label{xitrans}\ee
In terms of these variables, the reduced Hamiltonian coming from $H_1$ in \eqref{HamsonfM} becomes $\Re \tr(J) = \sum_{k=1}^n p_k$, while
$H_2 = \frac{1}{2} \Re \tr(J^2)$ yields the generalized spin Sutherland Hamiltonian $\cH_{\mathrm{spin}-2}$ \eqref{spinSuth2}
displayed in the Introduction.
\end{lemma}

\begin{proof}
By decomposing $J$ as $J=J^+ + J^-$ using \eqref{Gpm}, we can write
\be
(e^{-q} J e^q)_{ij} = (\cosh(q_i - q_j) J^+_{ij} - \sinh(q_i - q_j) J^-_{ij}) + (\cosh(q_i - q_j) J^-_{ij} - \sinh(q_i -q_j) J^+_{ij}).
\ee
Since the first two terms give the anti-Hermitian part $(e^{-q} J e^q)^+$, we obtain from \eqref{xilrdef}
\be
\xi^r_{ij} = -\cosh(q_i - q_j) \xi^l_{ij} - \sinh(q_i - q_j) J^-_{ij}.
\ee
This relation can be solved for $J^-_{ij}$ as
\be
J^-_{ij} = - \coth(q_i -q_j) \xi^l_{ij} -\xi^r_{ij}/\sinh(q_i-q_j), \quad \hbox{if}\quad i\neq j,
\ee
and $J_{ii}^- = p_i$ by definition.
Thus we have derived  \eqref{Jpar} and its easy to see that this yields a smooth bijection between
$(e^q, J)\in \fM^\reg_0$ and the set of quadruplets $(q, p, \xi^l, \xi^r)$ satisfying the conditions
stated by the lemma.
The transformation rule \eqref{xitrans} is equivalent to \eqref{Jtrans}.
 The forms of $\Re \tr(J)$ and $\frac{1}{2} \Re \tr(J^2)$ then follow
by straightforward evaluation.  To obtain the formula \eqref{spinSuth2}, one applies the identity
$2 \cosh x /\sinh^2 x = 1/\sinh^2(x/2) - 2/\sinh^2 x$.
\end{proof}

It turns out that under the reduced first Poisson bracket \eqref{redPB1} the $q_i, p_i$ form canonically
conjugate pairs, the Poisson brackets of the $\T^n$ invariant functions of $\xi^l, \xi^r$
are governed by the Lie--Poisson bracket of $\u(n) \oplus \u(n)$ reduced by the constraints
\eqref{xilr0}, and these two sets of variables decouple  under $\{\ ,\ \}_1^\red$.
This result can be obtained as a consequence of the symplectic reduction approach
adopted in \cite{FP2}. Alternatively,   we can directly perform the required  change of variables
in the formula \eqref{redPB1}.
 In order to make this paper self-contained, we present
the second method,   but relegate all computational details to Appendix \ref{app:B}.
Incidentally, we have verified that the two methods give the same result, which provides
an excellent check on our considerations.

Thus,  by using \eqref{Jpar},  we
parametrize $\fM_0^\reg$ by the quintets of variables
\be
(q,p, \xi^l_\perp, \xi^r_\perp, \xi_0)
\quad\hbox{where}\quad \xi_0\in \cG^+_0 \quad\hbox{and}\quad \xi^l= \xi^l_\perp + \xi_0,\,\,
\xi^r = \xi^r_\perp - \xi_0.
\label{5var}\ee
For any smooth, real function $F$ of these variables, we have the `partial gradients'
\be
d_{\xi^l_\perp} F \in \cG^+_\perp, \quad
d_{\xi^r_\perp}F \in \cG^+_\perp,\quad
d_{\xi_0} F \in \cG^+_0,
\label{5der}\ee
which are defined in the natural manner using the restriction of the pairing \eqref{E1} to $\cG^+_\perp$ and to
$\cG^+_0$ (remember that $\cG^+ = \u(n)$).
For arbitrary $\T^n$ invariant functions $f,h$ of the old variables $(e^q, J)$ we write
\be
f(e^q, J) = F(q,p, \xi^l_\perp, \xi^r_\perp, \xi_0)
\quad\hbox{and}\quad
h(e^q, J) = H(q,p, \xi^l_\perp, \xi^r_\perp, \xi_0),
\ee
and calculate $\{F,H\}_1^\red$ from the identity
\be
\{ F,H\}_1^\red(q,p,\xi^l_\perp, \xi^r_\perp, \xi_0) :=  \{f,h\}_1^\red(e^q,J).
\ee

\begin{proposition}\label{Prop:4.8}
Let $F,H$ be $\T^n$ invariant smooth functions on $\fM_0^\reg$,
parametrized  by the variables
\eqref{5var} that transform according to
\be
(q,p,\xi^l_\perp, \xi^r_\perp, \xi_0) \mapsto (q,p, \tau \xi^l_\perp \tau^{-1},
\tau \xi^r_\perp \tau^{-1}, \xi_0), \qquad \forall \tau \in \T^n.
\ee
In terms of these variables,  the first reduced Poisson bracket \eqref{redPB1} can be written as
\bea
&& \{F,H\}_1^\red(q,p, \xi^l_\perp, \xi^r_\perp, \xi_0)=
\sum_{i=1}^n \left( \frac{\partial F}{\partial q_i} \frac{\partial H}{\partial p_i}
- \frac{\partial H}{\partial q_i} \frac{\partial F}{\partial p_i}\right)
\label{2spinPB}\\
&& +
 \langle \xi^l_\perp + \xi_0, [ d_{\xi^l_\perp} F +
  d_{\xi_0} F , d_{\xi^l_\perp} H +  d_{\xi_0} H] \rangle
 + \langle \xi^r_\perp - \xi_0, [ d_{\xi^r_\perp} F
 , d_{\xi^r_\perp} H  ] \rangle.
\nonumber\eea
 \end{proposition}

\begin{remark}\label{Rem:4.9}
The $\T^n$ invariance of $H$ implies that
\be
\langle \xi^l_\perp, [d_{\xi_0} F, d_{\xi^l_\perp} H]\rangle +
\langle \xi^r_\perp, [d_{\xi_0} F, d_{\xi^r_\perp} H]\rangle =0.
\label{Tinv1}\ee
By using this identity, and its counterpart with $F$ and $H$ exchanged, one may write the terms containing $d_{\xi_0}$ in the second line
of \eqref{2spinPB} in many alternative ways.
Upon imposing the constraint $\xi^l =0$, and putting $\xi:= \xi^r_\perp$, $J$ in \eqref{Jpar} reproduces $J^-$ in \eqref{J-par}.
Setting also all derivatives with respect to $\xi_0$ and $\xi^l_\perp$ to zero, the Poisson bracket
\eqref{2spinPB} reproduces the first reduced Poisson bracket of the spin Sutherland model
\eqref{spinSuth1} described in Remark \ref{Rem:4.4}. Of course,  by setting $\xi^r=0$ instead of $\xi^l=0$ one reaches the same model.
\end{remark}

\begin{remark} \label{Rem:4.10}
Let us recall from Remark \ref{Rem:2.4} that the master system behind
the spin Sutherland model \eqref{spinSuth2} is a degenerate integrable system.
It is known \cite{J,Zung} that degenerate integrability
is generically preserved under Hamiltonian reduction.
The degenerate integrability of the model \eqref{spinSuth2} on generic symplectic leaves
of the first Poisson structure follows as a special case of  results of \cite{Res3}.
It would be nice to enhance these  results by explicitly exhibiting
the required number of independent constants of motion.
We here restrict ourselves to displaying a large
number of $\U \times \U$ invariant elements of the ring $\fA$ of unreduced constants
of motion, discussed in Remark \ref{Rem:2.4}, which descend to constants of motion
of the reduced system.
Namely, let $P$ be
an arbitrary product of non-negative powers of $J^+$ and $J^-$. That is, $P$ has the form
\be
P= (J^+)^{k_1} (J^-)^{k_2} (J^+)^{k_3} (J^-)^{k_4}\cdots ,
\ee
with non-negative integers $k_1$, $k_2$, $k_3$, $k_4$ etc.
On account of the transformation rules \eqref{cR} and \eqref{cL},
the real and imaginary parts of $\tr P$ are $\U\times \U$ invariant
elements of $\fA$.
Similarly,  using \eqref{tJdef}, one obtains $\U \times \U$ invariant unreduced constants of motion
by taking the trace of an arbitrary product of powers of
${\tilde J}^+$ and ${\tilde J}^-$.
Degenerate integrable systems are also known to be Liouville integrable under very general conditions
\cite{J}.
Considering the restriction of the model to generic symplectic leaves,  a construction
of sufficient number of constant of motion in involution can be found
in \cite{KLOZ}.
\end{remark}

\section{Conclusion}
\setcounter{equation}{0}

We here introduced a bi-Hamiltonian hierarchy 
on the cotangent bundle of the real Lie group $\GL$ and
analyzed its quotient with respect to the symmetry group $\U \times \U$.
We described the form of the compatible reduced Poisson brackets (Theorems \ref{Prop:3.4}
and  \ref{Prop:3.5}) as well as the bi-Hamiltonian vector fields
generated by the commuting reduced Hamiltonians (Propositions \ref{Th:3.6} and
\ref{Prop:3.8}).
We found that the restriction of the reduced bi-Hamiltonian hierarchy to a joint
Poisson subspace of its two Poisson brackets reproduces the hyperbolic spin Sutherland
hierarchy associated with the Hamiltonian \eqref{spinSuth1}.
In the general case, the reduced system was identified as a Sutherland model coupled
to two spin variables  according to the Hamiltonian \eqref{spinSuth2}.
Thus the commuting flows of the model that are generated by the spectral invariants
of the Lax matrix \eqref{Jpar} all admit bi-Hamiltonian description, which may be considered
as our main result.

The spin Sutherland interpretation arose from using
the variables $(q,p,\xi^l, \xi^r)$ instead of $(e^q,J)$, in which the first reduced
Poisson bracket takes the form displayed in Proposition \ref{Prop:4.8}.
It is in principle possible to present also the second Poisson bracket in these variables,
but the resulting formulae are not expected to have a transparent  structure.
At least on the subspace where $J$ is Hermitian and positive definite,  it should be possible
to construct an alternative parametrization that would allow $\tr(J)$ to be interpreted
as a spin Ruijsenaars--Schneider Hamiltonian.  In the corresponding trigonometric case,
such a  change of variables is known \cite{FeLMP}, and it permits one to recover the
spinless trigonometric Ruijsenaars--Schneider model on a special symplectic leaf of
the reduced second Poisson bracket. Thus we suspect that the spinless hyperbolic Ruijsenaars--Schneider
model should be found on a symplectic leaf of the second Poisson structure described in
Theorem \ref{Th:4.3}. However, we do not know how to characterize the symplectic leaves
of this Poisson structure. We encountered difficulties when trying to find them by
`analytic continuation' from the trigonometric to the hyperbolic case.
This poses a very interesting open problem for future work.
It is worth mentioning that, at least to our knowledge,
no derivation of the
the real, hyperbolic Ruijsenaars--Schneider model by Hamiltonian reduction is known at present,
as opposed to the real trigonometric model and its complex holomorphic counterpart,
for which several reduction treatments are available \cite{A,CF1,FK2,FR,Ob}.
It would be important to construct such a derivation,
and finding the symplectic leaves of the Poisson bracket \eqref{-PB2} could help to resolve this long-standing
 conundrum.

\bigskip
\bigskip
\begin{acknowledgements}
I wish to thank Maxime Fairon for several useful remarks on the manuscript.
This work was supported in part by the NKFIH research grant K134946.
\end{acknowledgements}

\appendix
\section{An explanation of the second Poisson bracket}\label{app:A}
\setcounter{equation}{0}

We below outline how the second Poisson bracket \eqref{PB2} arises by a change
of variables from Semenov-Tian-Shansky's Poisson bracket \cite{STS} on the Heisenberg double $G\times G$.
We will be brief since this explanation closely follows the appendix in \cite{FeAHP}.
However, note that in \cite{FeAHP} we considered holomorphic complex functions, while
 now we deal with real smooth functions.

We begin  by introducing a non-degenerate, invariant
bilinear form on the real Lie algebra $\cG \oplus \cG$ by
\be
\langle (X_1,X_2), (Y_1, Y_2) \rangle_2 :=
\langle X_1, Y_1 \rangle - \langle X_2, Y_2 \rangle,
\label{G1}\ee
where $(X_1,X_2)$ and $(Y_1,Y_2)$ are from $\cG \oplus \cG$, and \eqref{E1} is applied.
It is not difficult to see that $\cG \oplus \cG$
is the vector space direct sum of the
subalgebras
\be
\cG^\delta := \{ (X,X)\mid X \in \cG\}
\label{G2}\ee
and
\be
\cG^*:= \{ (r_+(X), r_-(X))  \mid X  \in \cG\}.
\label{G3}\ee
Recall that $r_\pm$ are defined in \eqref{rpm}, i.e., $r_+(X) = X_> + \half X_0$ and
$r_-(X) = - X_< - \half X_0$ for $X$ written as in \eqref{E2}
These are isotropic subalgebras, meaning that the bilinear form \eqref{G1} vanishes on them separately.
After identifying $\cG$ with $\cG^\delta$, we  can use the bilinear form
to take $\cG^*$ \eqref{G3} as a model of the dual space of $\cG$,
which explains the notation.

Let us define the linear operator $\rho$ on $\cG \oplus \cG$ by
\be
\rho:= \frac{1}{2} \left( P_{\cG^\delta} - P_{\cG^*}\right)
\ee
 using the projections
 $P_{\cG^\delta}$ onto $\cG^\delta$ and $P_{\cG^*}$  onto $\cG^*$
 associated with the vector space direct sum $\cG \oplus \cG = \cG^\delta + \cG^*$.
 It features in two well-known \cite{STS} Poisson brackets on $C^\infty(G\times G,\R)$.
  For $\cF\in C^\infty(G\times G, \R)$,
  the $\cG \oplus \cG$-valued left- and right-derivatives are determined by
\bea
&&\langle  \cD \cF(g_1,g_2), (X_1, X_2) \rangle_2: = \dt \cF(e^{t X_1} g_1, e^{t X_2} g_2),\nonumber\\
&& \langle \cD' \cF(g_1,g_2), (X_1, X_2)\rangle_2: = \dt \cF(g_1 e^{t X_1} , g_2 e^{t X_2}),
\label{G6}\eea
where $t\in \R$, $(X_1,X_2)$ runs over $\cG \oplus \cG$ and $(g_1,g_2)\in G\times G$.
With these notations, the two Poisson brackets read
\be
\{\cF, \cH\}_\pm := \langle \cD \cF, \rho \cD \cH \rangle_ 2 \pm  \langle \cD' \cF, \rho \cD' \cH \rangle_ 2.
\label{PBpm}\ee
 The minus bracket is called the Drinfeld double bracket, and the plus one the Heisenberg double bracket.
 The former makes $G\times G$ into a Poisson--Lie group,
 and the latter is symplectic in a neighbourhood of the identity \cite{STS}.

Now we introduce new variables in a neighbourhood of  $(\1_n, \1_n)\in G \times G$.
We need the connected subgroups of $G$ corresponding to the subalgebras in \eqref{cGdec1}.
These are denoted $G_>$,   $G_<$ and $G_0$, respectively, where $G_>$
contains the upper triangular
complex matrices in $G$ having $1$ in their diagonal entries, and
  $G_0$ contains the diagonal matrices in $G$.
Then the connected Lie subgroups of $G\times G$ associated
with the subalgebras $\cG^\delta$ \eqref{G2} and $\cG^*$ \eqref{G3} are
\be
G^\delta := \{ g_\delta  \mid
g_\delta:= (g,g),\, g \in G\}
\label{G4}\ee
and
\be
G^*= \left\{  \eta_*  \mid \eta_*:=\left(\eta_> \eta_0, ( \eta_0 \eta_<)^{-1}\right),\,
\eta_> \in G_>,\, \eta_0\in G_0,\, \eta_< \in G_<\right\}.
\label{G5}\ee
Since $\cG \oplus \cG = \cG^\delta + \cG^*$,
 there exist open neighbourhoods of the identity in $G\times G$ whose elements can be factorized uniquely as
 \be
 (g_1, g_2) = g_{\delta L} \eta_{*R}^{-1} = \eta_{*L} g_{\delta R}^{-1},
 \label{G8}\ee
 where $(g_{\delta L}, \eta_{*R})$ and $(g_{\delta R}, \eta_{*L})$ vary in corresponding open sets
 around the identity in $G^\delta \times G^*$.
 Let us write
 \be
 g_{\delta R} = (g,g)
 \quad\hbox{and}\quad
 \eta_{*R} = \left(\eta_> \eta_0, ( \eta_0 \eta_<)^{-1}\right),
 \ee
 and use the factorizations \eqref{G8} to define the map
 \be
 \psi: (g_1, g_2) \mapsto (g, J) \quad \hbox{with}\quad J:= \eta_> \eta_0^2 \eta_<.
 \label{psi}\ee
By suitably choosing its domain, $\psi$ gives a diffeomorphism between open
subsets of $G\times G$ containing $(\1_n,\1_n)$, but it is advantageous for us to regard its
image as a subset of $G\times \cG$.

The following statement can be verified by essentially the same calculations
that were presented in \cite{FeAHP}.

\begin{proposition}
Let $\cF$ and $\cH$ be smooth real functions on the domain of the local diffeomorphism $\psi$ \eqref{psi}.
Referring to \eqref{PBpm},  for $F:= \cF \circ \psi^{-1}$ and $H:= \cH \circ \psi^{-1}$ define
\be
\{ F , H\}_2 := \{ \cF, \cH\}_+ \circ \psi^{-1}.
\label{locPB2}\ee
Then $\{F,H\}_2$ \eqref{locPB2} yields a Poisson bracket for smooth functions defined locally around
$(\1_n, \1_n)\in G\times \cG$, and it has the explicit form displayed in \eqref{PB2}.
\end{proposition}

The proposition guarantees that the Jacobi identity holds for the restriction
of the second Poisson bracket \eqref{PB2} to an open set around $(\1_n,\1_n)$.
This implies that \eqref{PB2} gives a Poisson bracket  on $C^\infty(\fM,\R)$, too.
Indeed, as determined by  \eqref{PB2},
the Poisson brackets of the coordinate functions on $\fM$ provided by the matrix elements \eqref{gcomp}
and the components of $J$ are real-analytic functions on $\fM$, and thus the Jacobi identity holds
for them globally since we know from the proposition that it holds on an open set.

\begin{remark}
Incidentally, by writing $\nabla_1 F(g,J) = g d_1 F(g,J)$, where $d_1 F$ is defined similarly to
\eqref{E5} using that $G$ is an open subset of $\cG$, one may extend both Poisson brackets
 \eqref{PB1} and \eqref{PB2}
to  $C^\infty(\cG \times \cG,\R)$ as well.
\end{remark}

 \begin{remark}
 For the aficionados of Poisson--Lie groups,
we remark that
$\u(n)^\delta < \cG^\delta$ has the property that its annihilator inside $\cG^*$ \eqref{G3}
with respect to the pairing \eqref{G1} is a Lie subalgebra of $\cG^*$.
By applying the general theory \cite{STS}, one can trace back the closure  statements
of Lemma \ref{Lem:3.1} and Lemma \ref{Lem:3.2} to this property.
We gave direct proofs of these lemmas, and thus there is no need to elaborate this point.
\end{remark}

\section{A remark on the proof of Lemma \ref{lem:4.2}}\label{app:details}
\setcounter{equation}{0}

The goal of this appendix is to help those readers who wish to go through the details of the last
step of the proof of Lemma \ref{lem:4.2},  which requires the demonstration of the following identity:
\bea
&&X:= \langle (\nabla_2 f)^-, (\nabla_2' h)^-\rangle  +\langle (\nabla_2' f)^+, (\nabla_2 h)^+\rangle
 -\cE(f,h)  \nonumber\\
&&\quad\,\,\, =  4 \langle (d_2 f)^+, (J^+ (d_2 h)^+ J^-)^+ \rangle + 4 \langle (d_2 f)^-, (J^- (d_2h)^- J^+)^- \rangle,
 \label{X1}\eea
 where $\cE(f,h)$ stands for the 2 terms obtained by exchanging $f$ and $h$.
We note that $X$ as defined above can be expressed in the alternative form
\be
X= \langle [ d_2 f,J]^+, (J d_2 h + d_2 h J)^+ \rangle + \langle J d_2 f, d_2 h J \rangle  - \cE(f,h).
\label{X2}\ee
 This was also used in the proof of Theorem \ref{Prop:3.5}, and one can easily check it.
When spelling out the expression \eqref{X2},
it is convenient to write
\be
f^\pm := (d_2 f)^\pm \quad \hbox{and} \quad h^\pm := (d_2 h)^\pm.
\ee
With this notation, we have
\bea
[d_2 f, J]^+ &=&  [f^+, J^+] + [f^-, J^-],\nonumber\\
( J d_2 h  +  d_2 h J)^- &=& (J^+ h^+ + h^+ J^+) + (J^- h^- + h^- J^-),\\
 ( J d_2 h  +  d_2 h J)^+ &=& (J^- h^+ + h^+ J^-) + (J^+ h^- + h^- J^+).
\eea
To verify the subsequent
statements,
one only needs to use the cyclic property of the trace form \eqref{E1} and that $\cG^+$ is perpendicular to $\cG^-$.

\begin{lemma}\label{lem:X1}
The following identity holds:
\bea
 \langle J d_2 f, d_2 h J \rangle  - \cE(f,h)&=&
 \langle [f^+, h^+] + [f^-, h^-], J^+ J^- + J^- J^+ \rangle
\\
&+&   \langle [f^+, h^-] + [f^-,h^+],  (J^-)^2 + (J^+)^2 \rangle. \nonumber
\eea
\end{lemma}

\begin{lemma}\label{lem:X2}  The following identity holds:
\bea
\langle [ d_2 f,J]^+, (J d_2 h + d_2 h J)^+ \rangle - \cE(f,h)& =& 2 \langle f^+,  J^+ h^+ J^- - J^- h^+  J^+\rangle \\
& + & 2  \langle f^-, J^- h^- J^+ - J^+ h^- J^- \rangle \nonumber \\
&+& \langle [h^+, f^+] + [h^-, f^-], J^+ J^- + J^- J^+ \rangle \nonumber\\
&+&  \langle [h^-, f^+]   + [ h^+, f^-], (J^-)^2 + (J^+)^2 \rangle. \nonumber
\eea
\end{lemma}

Finally, the claimed formula \eqref{X1} follows by combining the statements displayed above.

\section{Proof of Proposition \ref{Prop:4.8}}\label{app:B}
\setcounter{equation}{0}

We introduce the notations
\be
S(x):= \sinh x,\quad C(x):= \cosh x,\quad \Q:= \ad_q,
\ee
and parametrize $J$ according to \eqref{Jpar}, i.e.,
\be
J = p - R(q) \xi^l_\perp -  S(\Q)^{-1} \xi^r_\perp - \xi^l,\qquad
\xi^l = \xi^l_\perp + \xi_0,
\label{Jparr}\ee
where $p=\diag(p_1,\dots, p_n)\in \cG^-_0$ and $R(q) = \coth \Q$ on $\cG_\perp$.
Note that $S(\Q)^{-1}$ denotes the inverse of the restriction of $S(\Q)$
on $\cG_\perp$, where $S(\Q)^{-1} = (1/S)(\Q)$.
We shall use that $R(q)$ and $S(\Q)$ both map $\cG^{\pm}_\perp$ to $\cG^{\mp}_\perp$ \eqref{cGperp},
respectively, and they vanish on $\cG_0$,
while $\cG^\pm$ are invariant subspaces of $C(\Q)$.

We put
\be
f(e^q,J) = F(q, p, \xi^l_\perp, \xi^r_\perp, \xi_0),
\qquad
h(e^q,J) = H(q, p, \xi^l_\perp, \xi^r_\perp, \xi_0).
\label{fF}\ee
This leads to
\be
(d_2 f)^- = d_p F + S(\Q) d_{\xi^r_\perp} F,
\ee
\be
(d_2 f)^+ = -  d_{\xi_0} F - d_{\xi^l_\perp} F + C(\Q) d_{\xi^r_\perp} F,
\ee
and
\be
\nabla_1 f = d_q F - \left[ R(q) S(\Q)^{-1} \xi^r_\perp + S(\Q)^{-2} \xi^l_\perp,
S(\Q) d_{\xi^r_\perp} F\right]^-_0.
\label{nab1}\ee
Here, $d_q F, d_p F$ belong to $\cG^-_0$,  and for the other derivatives see  \eqref{5der}.
We recall the formula
\bea
\{ f,h\}_1^\red(e^q,J) &=&  \langle \nabla_1 f, (d_2h)^-_0\rangle - \langle \nabla_1 h, (d_2 f)^-_0 \rangle \label{reddPB1}\\
&+&\langle R(q) [d_2 f, J]^+, (d_2 h)^- \rangle   - \langle R(q) [d_2 h, J]^+, (d_2 f)^- \rangle
\nonumber \\
 &+& \langle J^+, [ (d_2 f)^-, (d_2 h)^-] - [ (d_2 f)^+, (d_2 h)^+] \rangle ,\nonumber
\eea
which we have to rewrite in terms of $F$ and $H$ \eqref{fF}.
To make the subsequent equations shorter, we  denote
\be
F_q:= d_q F,\quad
F_p := d_pF, \quad F_0:=d_{\xi_0} F, \quad F^r_\perp :=d_{\xi^r_\perp} F,
\quad F^l_\perp :=d_{\xi^l_\perp} F,
\quad F^l:= F^l_\perp + F_0,
\ee
and will often omit the argument $q$, $\Q$  in $R(q)$, $S(\Q)$ etc.
We shall write $\cE(F,H)$ for any expression obtained by exchanging the roles of $F$ and $H$.

We start by inspecting
\be
\langle R [d_2 f, J]^+, (d_2 h)^- \rangle = \langle R [(d_2 f)^-, J^-], (d_2 h)^-\rangle +
 \langle R [(d_2 f)^+, J^+], (d_2 h)^- \rangle.
\ee

\begin{lemma}\label{prop:C.8}
We have the identity
\bea
&&
\langle R [(d_2 f)^-, J^-], (d_2 h)^-\rangle - \cE(f,h)
= \langle \xi^r_\perp, [F^r_\perp, H^r_\perp ]\rangle
+ \langle C\xi^l_\perp, [F^r_\perp, H^r_\perp]\rangle \nonumber
\\
&& \qquad\qquad \qquad \qquad  +  \left( \langle F_p, [ S^{-1} \xi^r_\perp, C H_\perp^r]
+ [R \xi^l_\perp, C H_\perp^r]\rangle  - \cE(F,H)\right).
\label{--id}
\eea
\end{lemma}
\begin{proof}
We first note that
\be
\langle R [(d_2 f)^-, J^-], (d_2 h)^-\rangle =  - \langle [(d_2f)^-, J^-], R (d_2h)^-\rangle,
\ee
and
\be
R (d_2h)^- = R (d_p H + S H^r_\perp )= C H^r_\perp.
\ee
Then we substitute $(d_2 f)^-$ and $J^-$,  which gives
\bea
&& \langle R [(d_2 f)^-, J^-], (d_2 h)^-\rangle - \cE(f,h)=
 \left( \langle  [ F_p +S F^r_\perp , S^{-1} \xi^r_\perp -p],  C H^r_\perp \rangle - \cE(F,H)\right)
  \nonumber \\
 && \qquad \qquad \qquad +
\left(
\langle F_p, [R \xi^l_\perp, C H^r_\perp ]\rangle +
\langle CH^r_\perp,  [ SF^r_\perp, R \xi^l_\perp] \rangle
-\cE(F,H)\right).
\label{d-d-}\eea
By expanding $S$ and $C$ in terms of exponentials, one can check that
\be
 [ S F^r_\perp , C H^r_\perp] - \cE(F,H) = S [F^r_\perp, H^r_\perp].
\ee
By combining this with the equality $\langle V, S(\Q) W \rangle = -\langle S(\Q)V, W\rangle$
($\forall V,W\in \cG$),
 we find that
\be
 \langle  [ F_p +S F^r_\perp , S^{-1} \xi^r_\perp -p ],  C H^r_\perp \rangle - \cE(F,H)=
\langle \xi^r_\perp, [F^r_\perp, H^r_\perp] \rangle
 + \left(  \langle F_p, [S^{-1} \xi^r_\perp, C H^r_\perp ]\rangle - \cE(F,H)\right),
 \ee
 and
\be
\langle [CH^r_\perp,  SF^r_\perp], R \xi^l_\perp \rangle -\cE(F,H)
=\langle [F^r_\perp, H^r_\perp], C \xi^l_\perp \rangle.
\ee
To get the last equality, we also used that $S(\Q) R(q) \xi^l_\perp  = C(\Q) \xi^l_\perp$.
Putting together these identities, the claim is proved.
\end{proof}

The other terms in the second and third lines of \eqref{reddPB1} can be spelled out
straightforwardly, and the result is summarized as follows.

\begin{lemma}
The following equalities hold.
First,
\be
\langle R [(d_2 f)^+, J^+], (d_2 h)^-\rangle - \cE(f,h)
=
 -2 \langle \xi^l, [CF^r_\perp, C H^r_\perp] \rangle
+ \left(\langle \xi^l, [F^l, C H^r_\perp] \rangle -  \cE(F,H)\right).
\ee
Second,
\be
\langle J^+, [(d_2f)^-, (d_2h)^-]\rangle = -\langle \xi^l, [S F^r_\perp, S H^r_\perp] \rangle
+\left( \langle F_p,  [\xi^l_\perp,S H^r_\perp] \rangle - \cE (F,H)\right).
\ee
Third,
\be
 -\langle J^+, [(d_2f)^+, (d_2h)^+]\rangle
 =
 \langle \xi^l, [F^l, H^l] + [CF^r_\perp, C H^r_\perp] \rangle
 + \left(\langle \xi^l, [ CH^r_\perp, F^l]\rangle - \cE(F,H)\right).
\ee
\end{lemma}

The next statement results by collecting the terms from the preceding two lemmas.

\begin{lemma}\label{lem:C.9}
Denote by $X$ the sum of the last two lines of \eqref{reddPB1}. Then we have
\bea
&&X=\langle \xi^l, [F^l, H^l] \rangle + \langle \xi^r, [F^r_\perp, H^r_\perp] \rangle \nonumber\\
&& \qquad +\left(\langle F_p, [R  \xi^l_\perp + S^{-1} \xi^r_\perp  , C H^r_\perp ] +
[\xi^l_\perp, S H^r_\perp]  \rangle  - \cE(F,H)\right)  \nonumber\\
&&\qquad + \langle C \xi^l, [ F^r_\perp, H^r_\perp ] \rangle - \langle \xi^l,  [ C F^r_\perp, C H^r_\perp ] + [ S F^r_\perp, S H^r_\perp ]
 \rangle.
 \label{Xeq}
\eea
\end{lemma}
To get \eqref{Xeq}, we used that $C \xi_0 = \xi_0$, and rewrote $\langle \xi^r_\perp , [F^r_\perp, H^r_\perp] \rangle$
coming from Proposition \ref{prop:C.8} as
\be
\langle \xi^r_\perp , [F^r_\perp, H^r_\perp] \rangle = \langle \xi^r, [F^r_\perp, H^r_\perp] \rangle
+  \langle C\xi_0, [F^r_\perp, H^r_\perp] \rangle \quad\hbox{with}\quad \xi^r = \xi^r_\perp - \xi_0.
\ee
By spelling out $C$ and $S$ in terms of $e^{\pm \Q}$ and using
$[e^{\Q } V, e^{\Q} W]=e^{\Q} [V,W]$ ($\forall V,W\in \cG$) we obtain
\be
 [ C F^r_\perp, C H^r_\perp ] + [ S F^r_\perp, S H^r_\perp ] = C [  F^r_\perp,  H^r_\perp ],
\ee
which implies that the third line in Lemma \ref{lem:C.9} vanishes.

\begin{lemma} \label{lem:C.8}
 The following equalities  hold:
\be
\langle \nabla_1 f, (d_2h)^-_0\rangle = \langle F_q,  H_p\rangle + \langle [R\xi_\perp^r + S^{-1} \xi_\perp^l, F^r_\perp], H_p\rangle,
\ee
and
\be
\langle H_p, [R\xi_\perp^l + S^{-1} \xi_\perp^r, C F^r_\perp]  + [\xi^l_\perp, S F^r_\perp ]\rangle
 =  \langle [R\xi_\perp^r + S^{-1} \xi_\perp^l, F^r_\perp], H_p\rangle.
\ee
\end{lemma}
\begin{proof}
We use the identity \eqref{nab1} to write
\bea
\langle \nabla_1 f, (d_2h)^-_0\rangle &=& \langle F_q - [ R S^{-1} \xi^r_\perp + S^{-2} \xi^l_\perp, SF^r_\perp], H_p \rangle
  \nonumber \\
& = & \langle F_q, H_p\rangle - \langle  R S^{-1} \xi^r_\perp + S^{-2} \xi^l_\perp, S [F^r_\perp, H_p] \rangle \nonumber\\
&  =& \langle F_q, H_p\rangle + \langle [ R \xi^r_\perp + S^{-1} \xi^l_\perp, F^r_\perp],  H_p \rangle.
\eea
Regarding the second equality, we have
\bea
&& \langle  H_p, [R\xi_\perp^l + S^{-1} \xi_\perp^r, C F^r_\perp]  + [\xi^l_\perp, S F^r_\perp ]\rangle \nonumber\\
&& \quad  =
\langle  R [ H_p, \xi_\perp^l]  + S^{-1} [H_p, \xi_\perp^r], C F^r_\perp \rangle  - \langle S  [H_p, \xi^l_\perp],  F^r_\perp \rangle \nonumber\\
&& \quad =
\langle  C^2 S^{-1} [ H_p, \xi_\perp^l]  + R [H_p, \xi_\perp^r],  F^r_\perp \rangle  - \langle   [H_p, S \xi^l_\perp],  F^r_\perp \rangle \nonumber\\
&& \quad =
\langle   H_p, [(S+ S^{-1}) \xi_\perp^l  +  R \xi_\perp^r,  F^r_\perp]  \rangle  - \langle   H_p, [ S \xi^l_\perp,  F^r_\perp] \rangle \nonumber\\
&& \quad = \langle   H_p, [  R \xi_\perp^r + S^{-1} \xi_\perp^l ,  F^r_\perp]   \rangle,
\eea
which finishes the proof.
\end{proof}

In terms of the new variables,  we have $\{F,H\}_1^\red(q,p,\xi^l_\perp, \xi^r_\perp, \xi_0):= \{f,h\}_1^\red(e^q,J)$.
The combination of Lemma \ref{lem:C.9}, where the third line of the formula was
shown to vanish,  with Lemma \ref{lem:C.8} gives
\be
\{ F,H\}_1^\red = \langle F_q,  H_p \rangle - \langle H_q, F_p \rangle
+\langle \xi^l, [F^l, H^l] \rangle + \langle \xi^r, [F^r_\perp, H^r_\perp] \rangle.
\ee
This is the same as \eqref{2spinPB}, and thus Proposition \ref{Prop:4.8} is proved.

\end{document}